\documentclass[11pt]{article}
\usepackage[T1]{fontenc}
\usepackage{lmodern}

\usepackage{amsthm}
\usepackage{graphicx} 
\usepackage{array} 

\usepackage{amsmath, amssymb, amsfonts, verbatim}
\usepackage{hyphenat,epsfig,subfigure,multirow}
\usepackage{algorithm,algpseudocode}
\usepackage{algorithmicx}
\usepackage{algpseudocode}

\usepackage[usenames,dvipsnames]{xcolor}

\usepackage{tcolorbox}
\tcbuselibrary{skins,breakable}
\tcbset{enhanced jigsaw}

\definecolor{DarkRed}{rgb}{0.5,0.1,0.1}
\definecolor{DarkBlue}{rgb}{0.1,0.1,0.5}

\usepackage{hyperref}
\hypersetup{
colorlinks=true,
pdfnewwindow=true,
citecolor=ForestGreen,
linkcolor=DarkRed,
filecolor=DarkRed,
urlcolor=DarkBlue
}

\usepackage{bm}
\usepackage{url}
\usepackage{xspace}
\usepackage[mathscr]{euscript}

\usepackage{mdframed}

\usepackage{cite}
\usepackage{enumitem}

\usepackage[margin=1in]{geometry}

\newtheorem{theorem}{Theorem}
\newtheorem{lemma}{Lemma}[section]

\newtheorem{cor}{Corollary}
\newtheorem{claim}[lemma]{Claim}

\newtheorem{conj}{Conjecture}

\newtheorem{definition}{Definition}

\newtheorem{problem}{Problem}
\newtheorem{rem}{Remark}

\newtheorem*{claim*}{Claim}
\newtheorem*{proposition*}{Proposition}
\newtheorem*{lemma*}{Lemma}
\newtheorem*{problem*}{Problem}

\newtheorem{mdresult}[theorem]{Theorem}

\newtheorem{mdinvariant}{Invariant}

\newcommand{\Erdos}{Erd\H{o}s}
\newcommand{\Renyi}{R\'{e}nyi}

\allowdisplaybreaks

\newcommand{\toShrink}{-.20cm}
\newcommand{\toShrinkEnu}{-.2cm}

\newcommand{\poly}{\mbox{\rm poly}}
\newcommand{\polylog}{\mbox{\rm  polylog}}

\DeclareMathOperator*{\Exp}{\ensuremath{{\mathbb{E}}}}
\DeclareMathOperator*{\Prob}{\ensuremath{\textnormal{Pr}}}
\renewcommand{\Pr}{\Prob}

\newenvironment{tbox}{\begin{tcolorbox}[
		enlarge top by=5pt,
		enlarge bottom by=5pt,
		 breakable,
		 boxsep=0pt,
                  left=4pt,
                  right=4pt,
                  top=10pt,
                  arc=0pt,
                  boxrule=1pt,toprule=1pt,
                  colback=white
                  ]
	}
{\end{tcolorbox}}

\newcommand{\LOCAL}{\ensuremath{\mathsf{LOCAL}}\xspace}
\newcommand{\CONGEST}{\ensuremath{\mathsf{CONGEST}}\xspace}
\newcommand{\CLIQUE}{\mathsf{CONGESTED}\text{-}\mathsf{CLIQUE}}
\newcommand{\ID}{{\operatorname{ID}}}
\newcommand{\mix}{\ensuremath{\tau_{\operatorname{mix}}}}

\makeatletter
\newcommand{\algmargin}{\the\ALG@thistlm}
\makeatother
\newlength{\whilewidth}
\algdef{SE}[parWHILE]{parWhile}{EndparWhile}[1]
  {\parbox[t]{\dimexpr\linewidth-\algmargin}{%
     \hangindent\whilewidth\strut\algorithmicwhile\ #1\ \algorithmicdo\strut}}{\algorithmicend\ \algorithmicwhile}%
\algnewcommand{\parState}[1]{\State%
  \parbox[t]{\dimexpr\linewidth-\algmargin}{\strut #1\strut}}

\def\vol{\mathrm{Vol}}

\def\degree{\mathrm{deg}}
\def\vlow{V_{\operatorname{low}}}

\def\ndel{n^{\delta}}
\def\ndell{n^{2\delta}}
\def\tp{\tilde{p}}
\def\tpi{\tilde{\pi}}
\def\trho{\tilde{\rho}}

\def\NG{N_{\mathcal{G}}}

\def\Remove#1{{\sf Remove}\text{-}#1}
\def\Split#1{{\sf Split}\text{-}#1}

\newcommand{\distance}{\operatorname{dist}}

\title{Distributed Triangle Detection via Expander Decomposition}

\author{Yi-Jun Chang\\ University of Michigan \\ \footnotesize \texttt{cyijun@umich.edu}\and
Seth Pettie\\ University of Michigan \\ \footnotesize \texttt{pettie@umich.edu} \and
Hengjie Zhang\\ IIIS, Tsinghua University\\
\footnotesize \texttt{zhanghj15@mails.tsinghua.edu.cn}
}

\date{}

\begin{document}
\maketitle
\begin{abstract}
We present improved distributed algorithms for
triangle detection and its variants in the $\CONGEST$ model. 
We show that Triangle Detection, Counting, 
and Enumeration can be solved in
$\tilde{O}(n^{1/2})$ rounds. 
In contrast, the previous state-of-the-art 
bounds for Triangle Detection and Enumeration were 
$\tilde{O}(n^{2/3})$ and $\tilde{O}(n^{3/4})$, respectively,
due to Izumi and LeGall (PODC 2017).

The main technical novelty in this work is a distributed graph partitioning algorithm. We show that in $\tilde{O}(n^{1-\delta})$ rounds we can partition the edge set of the network $G=(V,E)$ into three parts $E=E_m\cup E_s\cup E_r$ such that
\begin{itemize}
\item Each connected component induced by $E_m$ has 
minimum degree $\Omega(n^\delta)$ and conductance $\Omega(1/\polylog(n))$.  As a consequence the mixing time of a random walk within the component is $O(\polylog(n))$.
\item The subgraph induced by $E_s$ has arboricity at most $n^{\delta}$.
\item $|E_r| \leq |E|/6$.
\end{itemize}
All of our algorithms are based on the following generic framework, which we believe is of interest beyond this work. 
Roughly, we deal with the set $E_s$ by an algorithm that is efficient for low-arboricity graphs, 
and deal with the set $E_r$ using recursive calls.  
For each connected component induced by $E_m$, we are able to simulate 
$\CLIQUE$ algorithms with small overhead by applying a routing algorithm 
due to Ghaffari, Kuhn, and Su (PODC 2017) for high conductance graphs.
\end{abstract}

\thispagestyle{empty}
\setcounter{page}{0}

\newpage

\section{Introduction \label{se:intro}}

We consider Triangle Detection problems in distributed networks.
In the $\LOCAL$ model~\cite{Peleg00}, which has no limit on bandwidth,
all variants of Triangle Detection can be solved in exactly \emph{one} round of communication: every vertex $v$ simply announces its neighborhood $N(v)$ to all neighbors.  However, in models that take bandwidth into account, e.g., $\CONGEST$, Triangle Detection becomes significantly more complicated.
Whereas many graph optimization problems
studied in the $\CONGEST$ model are intrinsically ``global'' (i.e., require at least diameter time)~\cite{AgarwalRKP18,HuangNS17,GhaffariH16a,GhaffariH16b,KrinningerN18,Elkin17a,Elkin17b},
Triangle Detection is somewhat unusual in that it can,
in principle, be solved using only locally available information.

\paragraph{The $\CONGEST$ Model.}
The underlying distributed network
is represented as an undirected graph $G=(V,E)$, where each vertex corresponds to a computational device, and each edge corresponds to
a bi-directional communication link.
We assume each $v\in V$ initially knows some global parameters such as  $n=|V|$, $\Delta = \max_{v \in V} \deg(v)$, and $D = \text{diameter}(G)$.
Each vertex $v$ has a distinct $\Theta(\log n)$-bit identifier $\ID(v)$.
The computation proceeds according to synchronized \emph{rounds}.
In each round, each vertex $v$ can perform unlimited
local computation, and may send a distinct
$O(\log n)$-bit message to each of its neighbors.

Throughout the paper we only consider the randomized variant of $\CONGEST$.
Each vertex is allowed to generate unlimited local random bits,
but there is no global randomness.

\paragraph{The Congested Clique Model.}
The $\CLIQUE$ model is a variant of
$\CONGEST$ that allows all-to-all communication.
Each vertex initially knows its adjacent edges and
the set of vertex IDs, which we can assume w.l.o.g.~is $\{1,\ldots,|V|\}$.
In each round, each vertex transmits $n-1$ $O(\log n)$-bit
messages, one addressed to each vertex in the graph.

Intuitively, the $\CONGEST$ model captures two constraints in distributed computing: {\em locality} and {\em bandwidth},
whereas the $\CLIQUE$ model only focuses on the
{\em bandwidth} constraint.
This difference makes the two models behave very differently.
For instance, the {\em minimum spanning tree} (MST) problem can be solved in $O(1)$ rounds in the $\CLIQUE$~\cite{JurdzinskiN2018}, but its round complexity is $\tilde{\Theta}(D + \sqrt{n})$ in $\CONGEST$~\cite{PelegR2000,DasSarma2012}.

One of the main reasons that some problems can be solved efficiently in $\CLIQUE$ is due to the routing algorithm of Lenzen~\cite{Lenzen2013}. As long as each vertex $v$ is the source and the destination of at most $O(n)$ messages, we can deliver all messages in $O(1)$ rounds.
Using this routing algorithm~\cite{Lenzen2013} as a communication primitive, many parallel algorithms can be transformed to efficient $\CLIQUE$ algorithms~\cite{Censor2016}.
For example, consider the distributed matrix multiplication problem, where the input matrices are distributed to the vertices such that the $i$th vertex initially knows the $i$th row. The problem can be solved in the $\CLIQUE$  model in $\tilde{O}(n^{1/3})$ rounds over semirings, or $\tilde{O}(n^{1-(2/\omega) + o(1)}) = o(n^{0.158})$ rounds over rings~\cite{Censor2016}.

\paragraph{Distributed Routing in Almost Mixing Time.}
A {\em uniform lazy random walk} moves a token around an undirected graph by iteratively applying the following
process for some number of steps:
with probability $1/2$ the token stays at the
current vertex and otherwise it moves to a uniformly
random neighbor.
In a connected graph $G=(V,E)$,
the stationary distribution of a lazy random walk is
$\pi(u) = \deg(u)/(2|E|)$.
Informally, the {\em mixing time}
$\mix(G)$ of a connected graph $G$ is
the minimum number of lazy random walk steps needed
to get within a negligible distance of the stationary distribution.
Formally:

\begin{definition}[Mixing time~\cite{GhaffariKS17}]\label{def:mix}
Let $p_t^s(v)$ be the probability that after $t$ steps of
a lazy random walk starting at $s$, the walk lands at $v$.
The mixing time $\mix(G)$ is the minimum $t$ such that for all $s \in V$ and $v \in V$,
we have $|p_t^s(v) - \pi(v)| \leq \pi(v)/|V|$.
\end{definition}

Ghaffari, Kuhn, and Su~\cite{GhaffariKS17} proved that if each vertex $v$ is the source and the destination of at most $O(\deg(v))$ messages, then all messages can be routed to their destinations in $\mix(G) \cdot 2^{O(\sqrt{\log n \log \log n})}$ rounds. The $2^{O(\sqrt{\log n \log \log n})}$ factor has recently been improved~\cite{GhaffariL2018} to $2^{O(\sqrt{\log n })}$.
The implication of this result is that many problems that can be solved efficiently in the $\CLIQUE$ can also be solved efficiently in
$\CONGEST$, \emph{but only if $\mix(G)$ is small}.
In particular, MST can be solved in
$\mix(G) \cdot 2^{O(\sqrt{\log n})}$ rounds  in $\CONGEST$~\cite{GhaffariL2018}. This shows that the $\tilde{\Omega}(\sqrt{n})$ lower bound~\cite{PelegR2000,DasSarma2012}
can be bypassed in networks with small $\mix(G)$.

At this point, a natural question to ask is whether or not this line of research~\cite{GhaffariKS17,GhaffariL2018} can be extended to a broader class of graphs (that may have high $\mix(G)$), or even general graphs. The main contribution of this paper is to show that this is in fact doable, and based on this approach we improve the state-of-the-art algorithms for triangle detection, counting, and enumeration.

\paragraph{Graph Partitioning.}
It is
 well known  that {\em any} graph can be decomposed into connected components of
 conductance $\Omega(\epsilon/\log n)$ (and hence $\poly(\epsilon^{-1},\log n)$ mixing time)
 after removing at most an $\epsilon$-fraction of the edges~\cite{spielman2004nearly,PatrascuT07,AroraBS2015,Trevisan08}.
 Moshkovitz and Shapira~\cite{MoshkovitzS18} showed that this bound is essentially tight.
 In particular, removing any constant fraction $\epsilon$ of the edges, the remaining
 components have conductance at most $O((\log\log n)^2/\log n)$.

 A slightly weaker version of this graph partition can be constructed in near-linear time (for fixed $\epsilon$) in the sequential computation model~\cite{spielman2004nearly}. Their algorithm uses random walks to explore the graph locally to find a cut with edge sparsity $O(1/\log n)$.
 If the output cut is $S$, then the time spent
 is $\tilde{O}(\vol(S))$.\footnote{By definition, $\vol(S)=\sum_{v\in S}\deg(v)$.} By iteratively finding a sparse
 cut and removing it from the graph,
 in $\tilde{O}(|E|)$ time a graph partition is obtained in which all components have $\Omega(1/\polylog(n))$ conductance.

 This graph partition and the idea of local graph exploration have found many applications, such as solving linear systems~\cite{spielman2004nearly}, unique games~\cite{AroraBS2015,Trevisan08,RaghavendraS2010}, analysis of personalized PageRank~\cite{AndersenCL08},
 minimum cut~\cite{KawarabayashiT15}, and property testing~\cite{KumarSS2018,GoldreichR1999}.

 In this work, we show that a
 variant of this graph partition can be constructed
 efficiently in the $\CONGEST$ model.  The new twist
 is to partition the edge set in \emph{three} parts,
 rather than two (i.e., removed and remaining edges).

 \paragraph{Distributed Triangle Detection.} Many variants of the triangle
 detection problem have been studied in the literature~\cite{Censor2016,IzumiL17}.
 \begin{description}
 \item[Triangle Detection.] Each vertex $v$ reports a bit $b_v$, and $\bigvee_v b_v = 1$
 if and only if the graph contains a triangle.

 \item[Triangle Counting.] Each vertex $v$ reports a number $t_v$, and $\sum_v t_v$ is exactly the total number of triangles in the graph.

 \item[Triangle Enumeration.] Each vertex $v$ reports a list $L_v$ of triangles,
 and $\bigcup_v L_v$ contains exactly those triangles in the graph.

 \item[Local Triangle Enumeration.] It may be desirable that every triangle be reported
 by one of the three participating vertices.  It is required that $L_v$ only
 contain triangles involving $v$.
 \end{description}

 Dolev, Lenzen, and Peled~\cite[Remark 1]{DolevLP12} showed that Triangle
 Enumeration can be solved deterministically
 in $O(n^{1/3}/\log n)$ time in the $\CLIQUE$.
 Censor-Hillel et al.~\cite{Censor2016} presented an algorithm for
 Triangle Detection and Dounting in $\CLIQUE$ that takes
 $\tilde{O}(n^{1-(2/\omega) + o(1)}) = o(n^{0.158})$ time
 via a reduction to matrix multiplication.
Izumi and LeGall~\cite{IzumiL17} showed that in $\CONGEST$,
the Detection and Enumeration problems can be solved in $\tilde{O}(n^{2/3})$ and $\tilde{O}(n^{3/4})$ time, respectively.
They also proved that in both $\CONGEST$ and $\CLIQUE$,
the Enumeration problem requires $\Omega(n^{1/3}/\log n)$ time, improving
an earlier $\Omega(n^{1/3}/\log^3 n)$ bound of Pandurangan et al.~\cite{PanduranganRS18}.
Izumi and LeGall~\cite{IzumiL17} proved a large separation between the complexity
of the Enumeration and Local Enumeration problems.  If triangles
must be reported by a participating vertex, $\Omega(n/\log n)$ time
is necessary (and sufficient) in $\CONGEST/\CLIQUE$.
More generally, the lower bound on Local Enumeration
is $\Omega(\Delta/\log n)$ when the maximum degree is $\Delta$.

 In this paper, we show that Triangle Detection, Enumeration,
 and Counting can be solved in $\tilde{O}(n^{1/2})$ time in $\CONGEST$.
 This result is achieved by a combination of our new distributed graph partition
 algorithm, the multi-commodity routing of~\cite{GhaffariKS17,GhaffariL2018},
 and a randomized version of the $\CLIQUE$ algorithm for Triangle Enumeration of~\cite{DolevLP12,Censor2016}.
 We also show that when the input graph has high
 conductance/low mixing time,
 that Triangle Enumeration can be solved even faster, in $O(\mix(G)n^{1/3+o(1)})$ time.

\subsection{Technical Overview}
Consider a graph $G = (V,E)$. For a vertex subset $S$, we write $\vol(S)$ to denote $\sum_{v \in S} \deg(v)$. Note that by default the degree is with respect to the original graph $G$. We write $\bar{S} = V \setminus S$, and let $\partial(S) = E(S, \bar{S})$ be the set of edges $e = \{u,v\}$ with $u \in S$ and $v \in \bar{S}$.  The {\em sparsity} (or \emph{conductance}) of a cut $(S, \bar{S})$ is defined as $\Phi(S) = |\partial(S)| /  \min\{\vol(S), \vol(\bar{S})\}$.
The {\em conductance} $\Phi_G$ of a graph $G$
is the minimum value of $\Phi(S)$ over all vertex subsets $S$.

We have the following relation~\cite{JerrumS89} between the
mixing time $\mix(G)$ and conductance $\Phi_G$:
\[
\Theta\left(\frac{1}{\Phi_G}\right) \leq \mix(G) \leq \Theta\left(\frac{\log n}{\Phi_G^2}\right).
\]
In particular, if the inverse of the conductance
is $n^{o(1)}$, then the mixing time is also $n^{o(1)}$.

\paragraph{Our Graph Partition.}
We introduce a new, efficiently computable graph decomposition
that partitions the edge set into \emph{three} parts.

\begin{definition}\label{def:decomp}
An {\em $n^{\delta}$-decomposition} of a graph $G=(V,E)$
is a tripartition of the edge set $E=E_m\cup E_s\cup E_r$
satisfying the following conditions.
\begin{itemize}
\item [$(a)$] Each connected component induced by $E_m$ has
$O(\poly \log n)$ mixing time, and each vertex in the component has $\Omega(n^{\delta})$ incident edges in $E_m$. That is, for each vertex $v \in V$, either $\deg_{E_m}(v) = 0$ or $\deg_{E_m}(v) = \Omega(n^{\delta})$.

\item [$(b)$] $E_s=\bigcup_{v\in V}E_{s,v}$, where $E_{s,v}$ is a subset of
edges incident to $v$ and $|E_{s,v}|\leq n^{\delta}$.
We view $E_{s,v}$ as \emph{oriented} away from $v$.
The overall orientation on $E_s$ is \emph{acyclic},
which certifies that $E_s$ has arboricity\footnote{The arboricity of a graph is the minimum number $\alpha$ such that its edge set can be partitioned into $\alpha$ forests.}
at most $n^\delta$.
Each vertex $v$ knows $E_{s,v}$.

\item [$(c)$] $|E_r|\leq |E|/6$.
\end{itemize}
\end{definition}

Throughout the paper we assume $\delta \in (0,1)$ is a constant.
The main difference between our graph partition and the ones in other works~\cite{spielman2004nearly} is that we allow a set
$E_s$ that induces a low arboricity subgraph.
The purpose of having the set $E_s$ is to allow us to design an efficient $\CONGEST$ algorithm to construct the partition. In the sequential computation model, a common approach to find a graph partition is to iteratively find a vertex set $S$ with small $\Phi(S) = O(1/\polylog(n))$, and then
include the boundary edges $\partial(S)$
in the set $E_r$ and remove them from the current graph.
The number of iterations can be as high as
$\tilde{\Theta}(n)$ since we could have $|S| = \tilde{O}(1)$.

To reduce the number of iterations to at most $O(n^{1 - \delta})$, before we start to find $S$, we do a preprocessing step that removes low degree vertices in such a way that each vertex has degree at least $\Omega(n^{\delta})$ in the remaining graph. This guarantees that $|S| = \Omega(n^{\delta})$, and so the number of iterations can be upper bounded by $O(n^{1 - \delta})$, since the total number of vertices is $n$.

 \subsection{Additional Related Works}
 Drucker et al.~\cite{DruckerKO14} showed an $\Omega\left(\frac{n}{e^{\sqrt{\log n}} \log n}\right)$ lower bound for triangle detection in the {\em broadcast} $\CLIQUE$ model,
 where each vertex can only broadcast one message to all other vertices in each round.
 In the $\CONGEST$ model, lower bounds for finding a triangles and other motifs (subgraphs) has been studied in~\cite{AbboudCK2017,DruckerKO14,korhonenR2018,GonenO18}.
 The problem of detecting a $k$-cycle has an $\tilde{\Omega}(\sqrt{n})$ lower bound, for any even number $k \geq 4$~\cite{DruckerKO14,korhonenR2018}.
 Detecting a $k$-clique requires $\tilde{\Omega}(\sqrt{n})$ rounds for every $4 \leq k \leq \sqrt{n}$, and  $\tilde{\Omega}(\sqrt{n} / k)$ rounds for every $k \geq \sqrt{n}$~\cite{CzumajK18}.

 Any \emph{\underline{one-round}} algorithm for the {\em triangle membership problem} (each vertex decides whether it belongs to a triangle) requires messages of
 size $\Omega(\Delta \log n)$~\cite{AbboudCK2017}, which meets the trivial upper bound.
 The distributed triangle detection problem has also been studied in the property testing setting in the $\CONGEST$ model~\cite{Even17}.

 Das Sarma et al.~\cite{DasSarma15} first studied the distributed sparsest cut problem. Specifically,
 given two parameters $b$ and $\phi$, if there exists a cut of balance at least $b$ and conductance at most
$\phi$, their algorithm outputs a cut of conductance at most $\tilde{O}(\sqrt{\phi})$ in $\tilde{O}((n + (1/\phi))/b)$ rounds. This result was later improved by Kuhn and Molla~\cite{KuhnM15} to $\tilde{O}(D + 1/(b \phi))$.\footnote{Kuhn and Molla~\cite{KuhnM15} further claimed that the output cut of their algorithm has balance at least $b/2$, but this claim turns out to be incorrect (Anisur Rahaman Molla, personal communication, 2018).} Their algorithms are built upon techniques in~\cite{sarma2009sparse}.

The local graph clustering algorithm of Spielman and Teng~\cite{spielman2004nearly}
has been improved, both in terms of running time and
the quality of the cuts discovered;
see, e.g.,~\cite{AndersenCL08,AndersenGPT2016,SpielmanT13,KwokL12}.

\subsection{Organization}

In Section~\ref{se:partition} we present a new distributed algorithm
for partitioning a graph into expanding subgraphs and a low-arboricity subgraph.
A key subroutine for finding a sparse cut is described in
Section~\ref{se:sparsecut}.
Section~\ref{se:triangle} presents Triangle Enumeration algorithms
for both expanding graphs and general graphs.
We conclude in Section~\ref{sect:conclusion} with a
conjecture on the complexity of distributed graph partitioning.

\section{Algorithm for Graph Partitioning \label{se:partition}}

We first introduce some notation.
Let $\deg_H(v)$ be the degree of $v$ in the subgraph $H$,
or in the graph induced by edge/vertex set $H$.
Let $V(E^\ast)$ be the set of vertices induced by the edge
set $E^\ast \subseteq E$.
The {\em strong diameter} of a subgraph $H$ of $G$ is defined as $\max_{u,v \in H} \distance_H(u,v)$ and the {\em weak diameter} of  $H$  is $\max_{u,v \in H} \distance_G(u,v)$.

The goal of this section is to prove the following theorem.

\begin{theorem}\label{th:main}
Given a graph $G=(V,E)$ with $n=|V|$, we can find, w.h.p.,
an $n^{\delta}$-decomposition in $\tilde{O}(n^{1-\delta})$ rounds
in the $\CONGEST$ model.
\end{theorem}

The algorithm for Theorem~\ref{th:main} is based on repeated application of a black box algorithm $\mathcal{A}^\ast$, which is given a subgraph $G'=(V',E')$
of the original graph $G=(V,E)$, where $V' = V(E')$, $n'=|V'|$, and $m'=|E'|$.
In $\mathcal{A}^\ast$, vertices may halt the algorithm at different times.

\paragraph{Specification of the Black Box.}
The goal of $\mathcal{A}^\ast$ is, given $G'=(V',E')$,
to partition $E'$ into $E'=E_m'\cup E_s'\cup E_r'$ satisfying some conditions.
The edge set $E_m'$ is partitioned into $E_m' = \bigcup_{i=1}^t \mathcal{E}_i$.
We write $\mathcal{V}_i = V(\mathcal{E}_i)$ and
$\mathcal{G}_i = (\mathcal{V}_i, \mathcal{E}_i)$,
and define $S=V'\setminus \left( \bigcup_{i=1}^t \mathcal{V}_i \right)$.

\begin{description}
\item[(C1)] The vertex sets
$\mathcal{V}_1, \ldots, \mathcal{V}_t, S$ are disjoint and partition $V'$.
\item[(C2)] The edge set $E_s'$ can be decomposed as
$E_s'=\bigcup_{v\in S}E_{s,v}'$, where $E_{s,v}'$
is a subset of edges incident to $v$, viewed as oriented away from $v$.
This orientation is acyclic.
For each vertex $v$ such that $E_{s,v}'\neq \emptyset$,
we have $|E_{s,v}'|+ \deg_{E_m'}(v) \leq \ndel$.
Each vertex $v$ knows the set $E_{s,v}'$.
\item[(C3)] Consider a subgraph $\mathcal{G}_i=(\mathcal{V}_i,\mathcal{E}_i)$.
Vertices in $\mathcal{V}_i$ halt after the same
number of rounds, say $K$.
Exactly one of the following subcases will be satisfied.
\begin{enumerate}
    \item [{\bf (C3-1)}] All vertices in $\mathcal{V}_i$ have degree $\Omega(\ndel)$ in the subgraph $\mathcal{G}_i$,  each connected component of $\mathcal{G}_i$ has $O(\poly \log n)$ mixing time, and $K = O(\poly \log n)$. Furthermore, every vertex in $\mathcal{V}_i$ knows that they are in this sub-case.
    \item [{\bf (C3-2)}] $|\mathcal{V}_i|\leq n' -\tilde{\Omega}(K\ndel)$, and every vertex in $\mathcal{V}_i$ knows they are in this subcase.
\end{enumerate}
\item[(C4)] Each vertex $v \in S$ halts in $\tilde{O}(n'/\ndel)$ rounds.
\item[(C5)] The inequality
$E_r'\leq \Big(|E'| \log|E'| -\sum_{i=1}^t |\mathcal{E}_i| \log |\mathcal{E}_i|\Big)/
(6\log m)$ is met.
\item[(C6)] Each cluster $\mathcal{V}_i$ has a distinct identifier. When a vertex $v \in \mathcal{V}_i$ terminates, $v$ knows the identifier of  $\mathcal{V}_i$. If $v \in S$, $v$ knows that it belongs to $S$.
\end{description}

We briefly explain the intuition behind these conditions.
The algorithm $\mathcal{A}^\ast$ will be applied recursively to all
subgraphs $\mathcal{G}_i$ that have yet to satisfy the minimum degree
and mixing time requirements specified in Theorem~\ref{th:main} and Definition~\ref{def:decomp}.
Because vertices in different components halt at various times,
they also may begin these recursive calls at different times.

The goal of (C2) is to make sure that once a vertex $v$ has $E_{s,v}' \neq \emptyset$, the total number of edges added to $E_{s,v}$ cannot exceed $n^\delta$.
The goal of (C3) is to guarantee that the component size drops at a fast rate.
The idea of (C5) is that the size of $E_r'$ can be mostly charged to the number of the edges in the small-sized edge sets $\mathcal{E}_i$;
this is used to bound the size of $E_r$ of our graph partitioning algorithm.

Note that in general the strong diameter of a subgraph $\mathcal{G}_i$ can be much higher than the maximum running time of vertices in  $\mathcal{G}_i$, and it could be possible that $\mathcal{G}_i$ is not even a connected subgraph of $G$. However, (C6) guarantees that each vertex $v \in \mathcal{V}_i$ still knows that it belongs to $\mathcal{V}_i$. This property allows us to recursively execute  $\mathcal{A}^\ast$ on each subgraph $\mathcal{G}_i$.

\begin{lemma}\label{le:auxiliary}
There is an algorithm $\mathcal{A}^\ast$ that finds a partition $E'=E_m'\cup E_s'\cup E_r'$ meeting the above specification in the $\CONGEST$ model, w.h.p.
\end{lemma}

Assuming Lemma \ref{le:auxiliary},
we are now in a position to prove Theorem~\ref{th:main}.

\begin{proof}[Proof of Theorem~\ref{th:main}]
Let $\mathcal{A}^\ast$ be the algorithm for Lemma~\ref{le:auxiliary}.
Initially, we apply $\mathcal{A}^\ast$ with $G'=G$, and this returns a partition $E'=E_m'\cup E_s'\cup E_r'$.

For each  subgraph $\mathcal{G}_i$ in the partition output by an invocation of $\mathcal{A}^\ast$, do the following.
If $\mathcal{G}_i$ satisfies (C3-1), by definition it must have
$O(\poly \log n)$ mixing time, and all vertices in $\mathcal{G}_i$ have degree $\Omega(\ndel)$ in $\mathcal{G}_i$; we add the edge set $\mathcal{E}_i$ to the set $E_m$
and all vertices in $\mathcal{V}_i$ halt.
Otherwise we apply the algorithm recursively to $\mathcal{G}_i$,
i.e., we begin by applying $\mathcal{A}^\ast$ to $G' = \mathcal{G}_i$ to further partition its edges.
All recursive calls proceed in parallel,
but may begin and end at different times.
Conditions (C1) and (C6) guarantee that this is possible.
(Note that if $\mathcal{G}_i$ is disconnected,
then each connected component of $\mathcal{G}_i$
will execute the algorithm in isolation.)

Initially $E_r = \emptyset$ and $E_s = \emptyset$. After each invocation of $\mathcal{A}^\ast$, we update $E_r \gets E_r\cup E_r'$, $E_s \gets E_s\cup E_s'$, and $E_{s,u} \gets E_{s,u}\cup E_{s,u}'$ for each vertex $u$.

\paragraph{Analysis.} We verify that the three conditions of Definition~\ref{def:decomp} are satisfied.
First of all, note that each connected component of $E_m$ terminated
in (C3-1) must have $O(\poly \log n)$ mixing time, and all vertices in the component have degree $\Omega(\ndel)$ within the component.
Condition (a) of Definition~\ref{def:decomp} is met.
Next, observe that Condition (b) of Definition~\ref{def:decomp}
is met due to (C2). If the output of $\mathcal{A}^\ast$ satisfies that $E_{s,v}' \neq \emptyset$, then $|E_{s,v}|$ together with the number of remaining incident edges (i.e., the ones in $E_m'$) is less then $\ndel$. Therefore, $|E_{s,v}|$ cannot exceed $\ndel$, since only the edges in $E_m'$ that are incident to $v$ can be added to $E_{s,v}$ in future recursive calls.
Lastly, we argue that (C5) implies that Condition (c) of Definition~\ref{def:decomp} is satisfied.
Assume, inductively, that a recursive call on edge set $\mathcal{E}_i$
eventually
contributes at most $|\mathcal{E}_i|\log|\mathcal{E}_i|/(6\log m)$ edges
to $E_r$.  It follows from (C5) that the recursive call on edge set $E'$
contributes $|E'|\log|E'|/(6\log m)$ edges to $E_r$.  We conclude
that $|E_r| \le |E|\log|E|/(6\log|E|)=|E|/6$.

Now we analyze the round complexity. In one recursive call of $\mathcal{A}^\ast$, consider a component $\mathcal{G}_i$ in the output partition, and let $K$ be the running time of vertices in $\mathcal{V}_i$.
Due to (C3), there are two cases. If $\mathcal{G}_i$ satisfied (C3-1),
it will halt
in $K = O(\poly \log n)$ rounds.
Otherwise, (C3-2) is met, and we have $|\mathcal{V}_i|\leq n' -\tilde{\Omega}(K\ndel)$.
Let $v \in V$ be any vertex, and let $K_1, \ldots, K_z$ be the running times of all calls to
$\mathcal{A}^\ast$ that involve $v$.
(Whenever $v$ ends up in $S$ or in a component satisfying (C3-1) it halts permanently, so $K_1,\ldots,K_{z-1}$ reflect executions that satisfy (C3-2) upon termination.)
Then we must have $\sum_{i=1}^z K_i \leq \tilde{O}(n / \ndel) + O(\poly \log n) = \tilde{O}(n^{1-\delta})$. Thus, the whole algorithm stops within $\tilde{O}(n^{1-\delta})$ rounds.
\end{proof}

\subsection{Subroutines}
Before proving Lemma \ref{le:auxiliary}, we first introduce some helpful subroutines.
Lemma~\ref{le:highdiameter} shows that for subgraphs of sufficiently high strong diameter, we can find a sparse cut of the subgraph, with runtime proportional to the strong diameter.
Lemma~\ref{le:lowdegree} offers a procedure that removes a set of edges in such a way that the vertices in the remaining graph have high degree, and the removed edges form a low arboricity subgraph.
Lemma~\ref{le:lowconductance} shows that if a subgraph already has a low conductance cut, then we can efficiently
find a cut of similar quality.

All these subroutines are applied to a connected subgraph $G^\ast=(V^\ast,E^\ast)$ of the underlying network $G=(V,E)$, and the computation does not involve vertices outside of $G^\ast$. In subsequent discussion in this section, the parameters $n$ and $m$ are always defined as $n = |V|$ and $m = |E|$, which are independent of the chosen subgraph  $G^\ast$.

\begin{lemma}\label{le:largediameter}
Let $m$ and $D$ be two numbers.
Let $(a_1,\ldots,a_D)$ be a sequence of positive integers such that  $D\geq 48\log^2 m$ and $\sum_{i=1}^D a_i\leq m$.
Then there exists an index $j$ such that $j\in [D/4, 3D/4]$
and
\[
a_j\leq \frac{1}{12\log m}\cdot \min\left(\sum_{i=1}^{j-1} a_i, \;\, \sum_{i=j+1}^D a_i\right).
\]
\end{lemma}
\begin{proof} Define $S_k=\sum_{i=1}^k a_i$ to be the $k$th prefix sum. By symmetry, we may assume $S_{\lfloor D/2 \rfloor}\leq S_D-S_{\lfloor D/2 \rfloor}$, since otherwise we can reverse the sequence.
Scan each index $j$ from $D/4$ to $D/2$. If an index $j$ does not satisfy $a_j\leq \frac{1}{12 \log m} \cdot S_{j-1}$, then this implies that $S_j > S_{j-1}\left(1+\frac{1}{12\log m}\right)$. If no index $j\in[D/4,D/2]$ satisfies this condition then $S_{\lfloor D/2 \rfloor}$ is larger than
\[
S_{\lfloor D/4 \rfloor} \cdot \left(1+\frac{1}{12\log m}\right)^{D/4}
\ge
S_{\lfloor D/4 \rfloor} \cdot \left(1+\frac{1}{12\log m}\right)^{12\log^2 m}
\geq S_{\lfloor D/4 \rfloor}\cdot m,
\]
which is impossible since $\sum_{i=1}^D a_i\leq m$. Therefore, there must exist an index $j\in[D/4,D/2]$ such that $a_j\leq \frac{1}{12 \log m}\cdot S_{j-1} = \frac{1}{12 \log m}\cdot \sum_{i=1}^{j-1}a_i$.
By our assumption that $S_{\lfloor D/2 \rfloor}\leq S_D-S_{\lfloor D/2 \rfloor}$, we also have
 $a_j\leq \frac{1}{12 \log m}\cdot \min\left(\sum_{i=1}^{j-1}a_i,\;
 \sum_{i=j+1}^D a_i\right)$.
\end{proof}

\begin{lemma}[High Diameter subroutine]\label{le:highdiameter}
Let $G^\ast=(V^\ast,E^\ast)$ be a connected subgraph and
$x\in V^\ast$ be a vertex for which
$\tilde{D} = \max_{v\in V^\ast} \distance_{G^\ast}(x,v) \geq 48\log^2 m$.
Define $\vlow=\{v\in V^\ast \mid \degree_{G^\star}(v)\leq\ndel/2\}$.
Suppose there are no edges connecting two vertices in $\vlow$. Then we can find a cut $(C,\bar{C})$ of $G^\ast$ such that $\min(|C|,|\bar{C}|)\geq\frac{\tilde{D}}{32}n^{\delta}$ and $\partial(C)\leq \min(\vol(C),\vol(\bar{C}))/ (12\log m)$ in $O(\tilde{D})$ rounds deterministically in the $\CONGEST$ model. Each vertex in $V^\ast$ knows whether or not it is in $C$.
\end{lemma}
\begin{proof}
The algorithm is as follows.
First, build a BFS tree of  $G^\ast$ rooted at
$x \in V^\ast$ in $O(\tilde{D})$ rounds.
Let $L_i$ be the set of vertices of
level $i$ in the BFS tree,
and let $p_i$ be the number of edges $e=\{u,v\}$ such that $u\in L_i$ and $v\in L_{i+1}$.
We write
$L_{a..b} = \bigcup_{i=a}^b L_i$.
In $O(\tilde{D})$ rounds we can let the root
$x$ learn the sequence $(p_1, \ldots, p_{\tilde{D}})$.

Note that in a BFS tree, edges do not connect
two vertices in non-adjacent levels.
By Lemma~\ref{le:largediameter}, there exists an index $j\in[\tilde{D}/4,3\tilde{D}/4]$ such that $p_j\leq \frac{1}{12\log m}\cdot \min\left(\vol(L_{1..j}),\vol(L_{j+1..\tilde{D}})\right)$,
and such an index $j$ can be computed locally at the vertex $x$.

The cut is chosen to be $C=L_{1..j}$, so we have $\partial(C)\leq \min(\vol(C),\vol(\bar{C}))/ (12\log m)$. As for the second condition, due to our assumption in the statement of the lemma, for any two adjacent levels $L_i,L_{i+1}$, there must exist a vertex $v\in L_i\cup L_{i+1}$ such that $v\notin \vlow$.
By definition of $\vlow$, $v$ has more than $\ndel/2$ neighbors in $G^\ast$, and they are all within $L_{i-1..i+2}$. Thus, the number of vertices within any four consecutive levels must be greater than $\ndel/2$. Since $j\in[\tilde{D}/4,3\tilde{D}/4]$, we have
\[
\min(|C|,|\bar{C}|)\geq \frac{\tilde{D}}{4}/4\cdot \ndel/2\geq \frac{\tilde{D}}{32}n^{\delta}.
\]
To let each vertex in $V^\ast$ learn whether or not it is in $C$,  the root $x$ broadcasts the index $j$ to all vertices in $G^\ast$. After that, each vertex in level smaller than or equal to $j$ knows that it is in $C$; otherwise it is in $\bar{C}$.
\end{proof}

Intuitively, Lemma~\ref{le:lowdegree} says that after the removal of a subgraph of small arboricity (i.e., the edge set $E_s^\diamond$), the remaining graph (i.e., the edge set $E^\diamond$) has high minimum degree. The runtime is proportional to the number of removed vertices (i.e., $|V^\ast|-|V^\diamond|$) divided by the threshold $\ndel$.
Note that the second condition of Lemma~\ref{le:lowdegree} implies that $E_{s,v}^\diamond = \emptyset$ for all $v \in V^\diamond$.

\begin{lemma}[Low Degree subroutine]\label{le:lowdegree}
Let $G^\ast=(V^\ast,E^\ast)$ be a connected subgraph with strong diameter $D$. We can partition $E^\ast =E^\diamond \cup E_s^\diamond$ meeting the following two conditions.
\begin{itemize}
    \item[1.] Let $V^\diamond$ be the set of vertices induced by $E^\diamond$. Each $v \in V^\diamond$ has more than $\ndel/2$ incident edges in $E^\diamond$.
    \item[2.] The edge set $E_s^\diamond$ is further partitioned as $E_s^\diamond=\bigcup_{v\in V^\ast \setminus V^\diamond} E_{s,v}^\diamond$, where $E_{s,v}^\diamond$ is a subset of incident edges of $v$, and $|E_{s,v}^\diamond|\leq \ndel$. Each vertex $v$ knows $E_{s,v}^\diamond$.
\end{itemize}
This partition can be found in $O(D+(|V^\ast|-|V^\diamond|)/\ndel)$  rounds deterministically in the $\CONGEST$ model.
\end{lemma}

\begin{proof}
To meet Condition $1$, a naive approach is to iteratively ``peel off''  vertices that have degree at most $\ndel/2$, i.e., put all their incident edges in $E_s$, so long as any such vertex exists.
On some graphs this process requires $\Omega(n)$ peeling iterations.

We solve this issue by doing a batch deletion.
First, build a BFS tree of $G^\ast$ rooted at an arbitrary vertex
$x  \in V^\ast$.
We use this BFS tree to let $x$  count the number of vertices that have degree less than $\ndel$ in the remaining subgraph in $O(D)$ rounds.

The algorithm proceeds in iterations.
Initially we set $E^\diamond \gets E^\ast$ and $E_s^\diamond \gets \emptyset$.
In each iteration, we identify the subset $Z \subseteq V^\ast$
whose vertices have at most $\ndel$ incident edges in $E^\diamond$.
We orient all the $E^\diamond$-edges
touching $Z$ away from $Z$, if one endpoint is in $Z$,
or away from the endpoint with smaller $\ID$, if both endpoints are in $Z$.
Edges incident to $v$ oriented away from $v$ are added to $E_{s,v}^\diamond$
and removed from $E^\diamond$.
The root $x$ then counts the number $z=|Z|$ of such vertices via the BFS tree.
If  $z > \ndel/2$, we proceed to the next iteration; otherwise we terminate the algorithm.

The termination condition ensures that each vertex has degree at least
$(\ndel+1) - z > \ndel/2$, and so Condition 1 is met.
It is straightforward to see that
the set $E_s^\diamond$ generated by the algorithm meets Condition 2,
since for each $v$, we only add edges to $E_{s,v}^\diamond$ once,
and it is guaranteed that $|E_{s,v}^\diamond| \leq  \ndel$.
Tie-breaking according to vertex-$\ID$ ensures the orientation is acyclic.

Throughout the process, each time one vertex puts any edges into $E_s^\diamond$, it no longer stays in $V^\diamond$. Each iteration can be done in $O(D)$ time. We proceed to the next iteration only if there are more than  $\ndel/2$ vertices being removed from  $V^\diamond$. A trivial implementation can lead to an algorithm taking $O(D\left\lceil (|V^\ast|-|V^\diamond|)/\ndel\right\rceil))$ rounds.
The round complexity can be further improved to
$O(D+ (|V^\ast|-|V^\diamond|)/\ndel)$ by pipelining the iterations. At some point the root $x$ detects that iteration $i$ was the last iteration; in $O(D)$ time it broadcasts a message to all nodes instructing them to
roll back iterations $i+1,i+2,\ldots$, which have been executed speculatively.\end{proof}

The proof of the following lemma is deferred to Section~\ref{se:sparsecut}.
It is a consequence of combining Lemmas~\ref{le:correctness} and~\ref{le:Nibblerunningtime}.

\begin{lemma}[Low Conductance subroutine]\label{le:lowconductance}
Let $G^\ast=(V^\ast,E^\ast)$ be a connected subgraph with strong diameter $D$.
Let $\phi\leq 1/12$ be a number.
Suppose that there exists a subset $S\subset V^\ast$ satisfying
\[
\vol(S)\leq (2/3)\vol(V^\ast) ~~~\text{and}~~~ \Phi(S)\leq \frac{\phi^3}{19208\ln ^2(|E^\ast|e^4)}.
\]
Assuming such an $S$ exists,
there is a $\CONGEST$ algorithm
that finds a cut $C \subset V^\ast$
such that $\Phi(C)\leq 12\phi$
in $O(D + \poly (\log |E^\ast|, 1/\phi))$ rounds,
with failure probability $1/\poly(|E^\ast|)$.
Each vertex in $V^\ast$ knows whether or not it
belongs to $C$.
\end{lemma}

\subsection{Proof of Lemma~\ref{le:auxiliary}}

We prove Lemma~\ref{le:auxiliary} by presenting and analyzing a specific distributed algorithm, which makes use of the subroutines specified
in Lemmas~\ref{le:highdiameter}, \ref{le:lowdegree}, and \ref{le:lowconductance}.

Recall that we are given a subgraph with edge set $E'$ and must ultimately
return a partition of it into $E_m'\cup E_s' \cup E_r'$.
The algorithm initializes $E_m' \gets E'$, $E_s' \gets \emptyset$,
and $E_r' \gets \emptyset$.
There are two types of special operations.
\begin{description}
\item[Remove.] In an {\sf Remove} operation, some edges are moved from $E_m'$ to either $E_s'$ or $E_r'$. For the sake of a clearer presentation, each such operation is tagged \Remove{$i$}, for some index $i$.
\item[Split.] Throughout the algorithm we maintain a partition of the current set $E_m'$. In a {\sf Split} operation, the partition subdivided. Each such operation is tagged as \Split{$i$}, for some index $i$, such that \Split{$i$} occurs right after \Remove{$i$}.
\end{description}

Throughout the algorithm, we ensure that any part $E^\star$ of the partition of $E_m'$ has an identifier that is known to all members of $V(E^\star)$.
It is not required that each part forms a connected subgraph.
The partition at the end of the algorithm,  $E_m'=\bigcup_{i=1}^t \mathcal{E}_i$, is the output partition.

\paragraph{Notations.}
Since we treat $E_m'$ as the ``active'' edge set and $E_s'$ and $E_r'$ as repositories of removed edges, $\deg(v)$ refers to the degree of $v$ in the subgraph induced by the \emph{current} $E_m'$.
We write $\vlow=\{v\in V' \mid \degree(v)\leq \ndel\}$.

\paragraph{Algorithm.}
In the first step of the algorithm,
move each edge $\{u,v\}\in E_m'$ in the subgraph
induced by $\vlow$
to $E_{s,u}'$, where $\ID(u) < \ID(v)$ (\Remove{1}).
(Breaking ties by vertex-$\ID$ is critical to
keep the orientation acyclic.)

After that, $E_m'$ is divided into connected components. Assume these components are $G_1=(V_1,E_1)$, $G_2=(V_2,E_2), \ldots$, where $V_i = V(E_i)$.
Let $D_i$ be the depth of a BFS tree rooted at an arbitrary vertex in $G_i$.
In $O(D_i)$ rounds, the subgraph $G_i$ is assigned an identifier that is known to all vertices in $V_i$ (\Split{1}).
Note that this step is done in parallel for each $G_i$, and the time for this step is different for each $G_i$.
From now on there will be no communication between different subgraphs in $\{G_1, G_2, \ldots \}$, and we focus on one specific subgraph $G_i$ in the description of the algorithm.

Depending on how large $D_i$ is, there are two cases.
If $D_i \geq 48\log^2m$, we go to Case 1,
otherwise we go to Case 2.

\paragraph{Case 1:} In this case, we have
$D_i \geq 48\log^2 m$.
Since there are no edges connecting two vertices in $\vlow$, we can apply the High Diameter subroutine, Lemma~\ref{le:highdiameter}, which finds a cut $(C,\bar{C})$ of $G_i$ such that $\min(|C|,|V_i \setminus C|)\geq\frac{D_i}{32}n^{\delta}$ and $\partial(C)\leq \min(\vol(C),\vol(V_i \setminus C))/ (12\log m)$ in $O(D_i)$ rounds.
Every vertex in $V_i$ knows whether it is in $C$ or not. All edges of the cut $(C,\bar{C})$ are put into $E_r'$ (\Remove{2}). Then $E_i$ splits into two parts according to the cut  $(C,\bar{C})$ (\Split{2}).
After that, all vertices in $V_i$ terminate.
(Observe that the part containing the BFS tree root is connected, but the other part is not necessarily connected.)

\paragraph{Case 2:} In this case,
we have $D_i \leq 48 \log^2 m$.
Since $G_i=(V_i,E_i)$ is a small diameter graph,
a vertex $v \in V_i$ is able broadcast a message to all
vertices in $V_i$ very fast.
We apply the Low Degree subroutine, Lemma~\ref{le:lowdegree}, to obtain a partition $E_i=E^\diamond \cup {E}_s^\diamond$.
We add all edges in ${E}_s^\diamond$ to $E_s'$ in such a way that $E_{s,v}' \gets E_{s,v}' \cup {E}_{s,v}^\diamond$ for all $v \in V_i \setminus V^\diamond$, where $V^\diamond = V(E^\diamond)$  (\Remove{3}).

After removing these edges, the remaining edges of $E_i$ are divided into several connected components, but all remaining vertices have degree larger than $\ndel/2$.
Assume these connected components are $G_{i,1}=(V_{i,1},E_{i,1})$, $G_{i,2}=(V_{i,2},E_{i,2})$, $\ldots$.
Let $D_{i,j}$ be the depth of the BFS tree from an arbitrary root vertex in $G_{i,j}$.
In $O(D_{i,j})$ rounds we compute such a BFS tree
and assign an identifier that is known to all vertices in
$V_{i,j}$ (\Split{3}).
That is, the remaining edges in $E_i$ are partitioned into $E_{i,1}$, $E_{1,2}$, $\ldots$.

In what follows, we focus on one subgraph $G_{i,j}$
and proceed to Case 2-a or Case 2-b.

\paragraph{Case 2-a:} In this case, $D_{i,j}\geq 48\log^2 m$.
The input specification of the High Diameter subroutine (Lemma~\ref{le:highdiameter}) is satisfied, since every vertex has degree larger than $\ndel/2$. We apply the High Diameter subroutine to $G_{i,j}$. This takes $O(D_{i,j})$ rounds. This case is similar to Case 1, and we do the same thing as what we do in Case 1, i.e., remove the edges in the cut found by the subroutine (\Remove{4}),  split the remaining edges (\Split{4}), and then all vertices in $V_{i,j}$ terminate.

\paragraph{Case 2-b:} In this case,
$D_{i,j} \leq 48 \log^2 m$.
Note that every vertex has degree larger than $\ndel/2$, and $G_{i,j}$ has small diameter. What we do in this case is to test whether $G_{i,j}$ has any low conductance cut; if yes, we will split $E_{i,j}$ into two components. To do so, we apply the Low Conductance subroutine, Lemma~\ref{le:lowconductance}, with $\phi=\frac{1}{144\log m}$.
Based on the result, there are two cases.

\paragraph{Case 2-b-i:} The subroutine finds a set of vertices $C$ that $\Phi(C)\leq 12\phi=\frac{1}{12\log m}$, and every vertex knows whether it is in $C$ or not.
We move $\partial(C)$ to $E_r'$ (\Remove{5}), and then split the remaining edges into two edge sets according to the cut $(C, \bar{C})$ (\Split{5}). After that, all vertices in $V_{i,j}$ terminate.

\paragraph{Case 2-b-ii:} Otherwise, the subroutine does not return a subset $C$, and it means with probability at least $1 - 1/\poly(|E_{i,j}|) = 1 - 1/\poly(n)$, there is no cut $(S,\bar{S})$ with conductance less than $\frac{\phi^3}{19208\ln^2( |E_{i,j}| e^4)} = \Theta(\log^{-5}m)$.
Recall the relation between the mixing time
$\mix(G_{i,j})$ and the conductance $\Phi=\Phi_{G_{i,j}}$: $\Theta(\frac{1}{\Phi}) \leq \mix(G_{i,j}) \leq \Theta(\frac{\log |V_{i,j}|}{\Phi^2})$~\cite{JerrumS89}.
Therefore, w.h.p., $G_{i,j}$ has $O(\poly \log n)$ mixing time. All vertices in $V_{i,j}$ terminate without doing anything in this step.

Note that in the above calculation, we use the fact that every vertex in $V_{i,j}$ has degree larger than $\ndel/2$ in $G_{i,j}$, and this implies that $|V_{i,j}| = \Omega(\ndel)$ and $|E_{i,j}| = \Omega(\ndell)$, and so $\Theta(\log m) = \Theta(\log n) = \Theta(\log |E_{i,j}|) = \Theta(\log |V_{i,j}|)$.

\paragraph{Analysis.}
We show that the output of
$\mathcal{A}^\ast$ meets its
specifications (C1)--(C6).
Recall that
$E_m'=\bigcup_{i=1}^t \mathcal{E}_i$ is the final partition of the edge set $E_m'$ when all vertices terminate. Once an edge is moved from $E_m'$
to either $E_r'$ or $E_s'$,
it remains there for the rest of the computation.
Condition (C1) follows from the fact that each time we do a split operation, the induced vertex set of each part is disjoint. Condition (C6) follows from the fact that each vertex knows which part of $E_m'$ it belongs to after each split operation. In the rest of this section, we prove that the remaining conditions are  met.

\begin{claim}
Condition (C2) is met.
\end{claim}
\begin{proof}
Note that only \Remove{1} and \Remove{3} involve $E_s'$.
 In \Remove{1}, any $E_{s,u}'$ that becomes non-empty must have had
 $u\in \vlow$, so $\deg(u)\le \ndel$ before \Remove{1}, and
 therefore $|E_{s,u}'| + \deg(u) \leq \ndel$ after \Remove{1}.
In \Remove{3}, the Low Degree subroutine of Lemma~\ref{le:lowdegree} computes a partition $E_i=E^\diamond \cup {E}_s^\diamond$,
and then we update $E_{s,u}' \gets E_{s,u}' \cup {E}_{s,u}^\diamond$ for all $u \in V_i \setminus V^\diamond$.
By Lemma~\ref{le:lowdegree}, for any $u$ such that ${E}_{s,u}^\diamond \neq \emptyset$, we have $|{E}_{s,u}^\diamond|\leq \ndel$, and $u \notin V^\diamond$, where $V^\diamond$ is the vertex set induced by the remaining edge set $E^\diamond$. In other words, once $u$ puts at least one edge  into $E_{s,u}'$, we have $\deg(u)=0$ after \Remove{3}.
\end{proof}

\begin{claim}
Conditions (C3) and (C4) are met.
\end{claim}
\begin{proof}
We need to verify that in each part of the algorithm, we either spend only $O(\poly \log n)$ rounds, or the size of the current component shrinks by $\tilde{\Omega}(n^{\delta})$ vertices \emph{per round}.

After removing all edges in the subgraph induced by $\vlow$, the rest of $E'$ is partitioned into
connected components
$\mathcal{E}_1,\mathcal{E}_2,\ldots$.
Consider one such component $\mathcal{E}_i$,
and suppose it goes to Case 1.
We find a sparse cut $(C,\bar{C})$, and
moving $\partial(C)$ to $E_r'$ breaks
$\mathcal{E}_i$ into
$\mathcal{E}_i^1$ and $\mathcal{E}_i^2$.
By Lemma~\ref{le:highdiameter},
we have  $\min(|C|,|\bar{C}|) \geq \frac{D_i}{32}\ndel$, so the size of both $V(\mathcal{E}_i^1)=C$ and
$V(\mathcal{E}_i^2) = \bar{C}$ are at most $|V(\mathcal{E}_i)|-\frac{D_i}{32}\ndel
\leq n' - \Omega(D_i)\ndel$.
Since the running time for each vertex in $V(\mathcal{E}_i^1)$ and $V(\mathcal{E}_i^2)$
is $O(D_i)$, the condition (C3-2) is met.

Now suppose that $\mathcal{E}_i$ goes to Case 2.
Note that the total time spent before it reaches Case 2 is $O(D_i)= \poly \log n$.
In Case 2 we execute the Low Degree subroutine of Lemma~\ref{le:lowdegree}, and let
the time spent in this subroutine be $\tau$.
By Lemma~\ref{le:lowdegree}, it is either the case that (i) $\tau = O(D_i)$
or
(ii) the remaining vertex set $V^\diamond$ satisfies $|V(E_i)| - |V^\diamond| = \Omega(\tau \ndel)$.
In other words, if we spend too much time (i.e., $\omega(D_i)$) on this subroutine, we must lose $\Omega(\ndel)$ vertices per round.

After that, $\mathcal{E}_i$ is split into
$\mathcal{E}_{i,1}$, $\mathcal{E}_{i,2}$, $\ldots$.
We consider the set $\mathcal{E}_{i,j}$.
If $\mathcal{E}_{i,j}$ goes to Case 2-a,
then the analysis is the same as that in Case 1,
and so (C3-2) is met.

Now suppose that $\mathcal{E}_{i,j}$ goes to Case 2-b. Note that the time spent during the Low Conductance subroutine of Lemma~\ref{le:lowconductance} is
$O(\poly \log n)$. Suppose that a low conductance cut $(C,\bar{C})$ is found (Case 2-b-i). Since the cut has conductance less than $\frac{1}{12\log m}$, by the fact that every vertex has degree higher than $\ndel/2$, we must have $\min(|C|, |\bar{C}|) = \Omega(\ndel)$. Assume $\mathcal{E}_{i,j} \setminus \partial(C)$ is split into $\mathcal{E}_{i,j}^1$ and $\mathcal{E}_{i,j}^2$.
The size of both $V(\mathcal{E}_{i,j}^1)$ and $V(\mathcal{E}_{i,j}^2)$
must be at most $|V(\mathcal{E}_{i,j})|- \Omega(\ndel)$.
Thus, (C3-2) holds for both parts
$\mathcal{E}_{i,j}^1$ and $\mathcal{E}_{i,j}^2$.

Suppose that no cut $(C,\bar{C})$ is found (Case 2-b-ii). If the running time $K$ among vertices in $V_{i,j}$ is $O(\poly \log n)$, then (C3-1) holds.
Otherwise, we must have $|{V}_{i,j}|\leq n' -\tilde{\Omega}(K\ndel)$ due to the Low Degree subroutine, and so (C3-2) holds.

Condition (C4) follows from the the above proof of (C3), since for each part of the algorithm, it is either the case that
(i) this part takes $O(\poly \log n)$ time,
or
(ii) the number of vertices in the current subgraph is reduced by $\tilde{\Omega}(\ndel)$ \emph{per round}.
\end{proof}

\begin{claim}
Condition (C5) is met.
\end{claim}
\begin{proof}
Condition (C5) says that after the algorithm $\mathcal{A}^\ast$ completes, $|E_r'|\leq f$, where
\[
f=\left(|E'| \log|E'| -\sum_{i=1}^t |\mathcal{E}_i| \log |\mathcal{E}_i|\right)/ (6\log m).
\]
We prove the stronger claim that this inequality holds at all times
w.r.t.~the current edge partition
$\mathcal{E}_1\cup \cdots\cup\mathcal{E}_t$ of $E_m'$.
In the base case this is clearly true, since $t=1$ and
$E' = E_m' = \mathcal{E}_1$ and
$E_r'=\emptyset$.  Moving edges from $E_m'$ to $E_s'$
increases $f$ and has no effect on $E_r'$, so we only have to consider the movement of edges from $E_m'$ to $E_r'$.
Note that this only occurs in \Remove{$i$} and \Split{$i$}, for $i \in \{2,4,5\}$, where in these operations we find a cut
$(C,\bar{C})$ and split one of the parts
$\mathcal{E}_j$ according to the cut.
In all cases we have
\[
|\partial(C)| \leq \frac{\min(\vol(C),\vol(\bar{C}))}{12\log m}.
\]
Suppose that removing $\partial(C)$ splits $\mathcal{E}_j$ into $\mathcal{E}_j^1$ and $\mathcal{E}_j^2$,
with $|\mathcal{E}_j^1| \le |\mathcal{E}_j^2|$
and $C=V(\mathcal{E}_j^1)$.
We bound the change in $|E_r'|$ and $f$ separately.  Clearly
\begin{align*}
\Delta|E_r'| &= |\partial(C)| \leq \frac{2|\mathcal{E}_j^1| + \partial(C)}{12\log m} \leq \frac{|\mathcal{E}_j^1|}{6\log m} + \frac{\partial(C)}{12\log m}.
\intertext{and}
\Delta f &= \frac{1}{6\log m}\cdot \left(|\mathcal{E}_j|\log|\mathcal{E}_j| - \sum_{k\in\{1,2\}} |\mathcal{E}_j^k|\log|\mathcal{E}_j^k|\right)\\
    &\ge \frac{1}{6\log m}\cdot \left( |\mathcal{E}_j^1|\log(|\mathcal{E}_j|/|\mathcal{E}_j^1|) + \partial(C)\log|\mathcal{E}_j|\right)\\
    &> \Delta|E_r'|  &  \mbox{(Because $|\mathcal{E}_j^1| < |\mathcal{E}_j|/2$.)}
\end{align*}
Thus, $|E_r'| \leq f$ also holds after \Remove{$i$} and \Split{$i$},
for $i\in\{2,4,5\}$.
\end{proof}

\section{Algorithm for Finding a Sparse Cut}\label{se:sparsecut}

Recall the in our decomposition routine, we search for a sparse cut in a \emph{subgraph} $G^\ast = (V(E^\ast), E^\ast)$ of $G$.
To simplify notation, we use $n=|V(G^\ast)|$ and $m=|E(G^\ast)|$ to be the number of vertices and 
edges in the subgraph.
In this section we prove Lemma~\ref{le:lowconductance}, which concerns 
an efficient distributed analogue of Spielman and Teng's~\cite{spielman2004nearly,SpielmanT13} \texttt{Nibble} routine.

Many existing works~\cite{sarma2009sparse,spielman2004nearly,AndersenCL08,KuhnM15} have shown that looking at the distribution of 
random walks is a good approach to finding a sparse cut. 
The basic idea is to first sample a source vertex $s$ according to the degree distribution, i.e., 
the probability that $v$ is sampled is $\deg(v) / (2m)$, and do a lazy random walk from $s$. 
Assume there is a sparse cut $S$ with conductance $\Phi(S)$, and $\vol(S)\leq\vol(V)/2$. If $s \in S$, then the probability distribution of the random walk will be mostly confined to $S$ within the initial $t_0=O(\frac{1}{\Phi(S)})$ steps. 
A common way to utilize this observation is to sort the vertices $(v_1, \ldots, v_n)$ 
in decreasing order of their random walk probability,  
and it is guaranteed that for some choice of $j$, the subset 
$C = \{v_1, \ldots, v_j\}$ is a sparse cut that is approximately as good as $S$. 

The papers~\cite{DasSarma15,KuhnM15} adapted this approach to the $\CONGEST$ model.
If the cut $S$ satisfies that $b \cdot 2|E| \leq \vol(S)$
(\emph{i.e., $S$ has balance $b$}), then a cut $C$ satisfying 
$\Phi(C)  = O(\sqrt{\Phi(S)\log n})$ 
can be found in $\tilde{O}(D + 1/(b\Phi(S)))$ rounds.
The algorithm is inefficient when $1/b = \Theta(|E| / \vol(S))$ is large.
The main source of this inefficiency is that if we sample a vertex 
$s$ according to the degree distribution, then the probability that $s \in S$ is only $O(b)$.
This implies that we have to calculate many random walk distributions 
before we find a desired sparse cut.
If we calculate these random walk distributions \emph{simultaneously}, 
then \emph{we may suffer from a huge congestion issue}.

Spielman and Teng~\cite{spielman2004nearly} show that a {\em random walk} distribution with \emph{truncation} (rounding a probability to zero when it becomes too small) can reveal a sparse cut, provided the starting
vertex of the random walk is good. 
The main contribution of this section is a proof that the Spielman-Teng method 
for finding cuts of conductance roughly $\phi$
can be implemented in $\poly(\phi^{-1},\log n)$ time
in the $\CONGEST$ model, i.e., 
with no dependence on the balance parameter $b$.

\paragraph{Terminology.} We first review some definitions and results from Spielman and Teng~\cite{spielman2004nearly}.
Let $A$ be the adjacency matrix of the graph $G$. 
We assume 1-1 correspondance between $V$ and $\{1, \ldots, n\}$.
In a \textit{lazy random walk}, the walk stays at the current vertex with probability $1/2$ and otherwise moves to a random neighbor of the current vertex.
The matrix realizing this walk can be expressed as
 $T = (A D^{-1} + I) / 2$,
  where $D$ is the diagonal matrix with
  $(d (1), \dotsc , d (n))$ on the diagonal, and
  $d(i) = \degree(i)$.
  
Let $p_t^v$ be the probability distribution of the lazy random walk that begins at $v$ and walks for $t$ steps.
In the limit, as $t\rightarrow \infty$, 
$p_t(x)$ approaches $d(x)/2m$, so it is natural to measure 
$p_t(x)$ \emph{relative} to this baseline.
\[
\rho_t(x)=p_t(x)/d(x),
\]
Define $\pi_t$ to be the permutation that sorts $V=\{1,\ldots,n\}$ in decreasing order of $\rho_t$-values, breaking ties by vertex ID.
(We never actually compute $\pi_t$.  To implement our algorithms,
it suffices that given
$\rho_t(u),\rho_t(v),\ID(u),\ID(v)$, we can determine
whether or not $u$ precedes $v$ according to $\pi_t$.)
\[
\rho_t(\pi_t(i)) \geq \rho_t(\pi_t(i+1)), \text{ for all } i.
\]

Let $p$ be a distribution on $V$.  
The truncation operation $[p]_\epsilon$ 
rounds $p(x)$ to zero if it falls below a threshold that depends on $x$.
\[
  [p]_{\epsilon} (x)
=
\begin{cases}
  p (x) & \text{if $p (x) \geq 2 \epsilon d(x)$,}\\
  0 & \text{otherwise}.
\end{cases}
\]

The truncated random walk starting at vertex $v$ is defined as follows.
In subsequent discussion we may omit $v$ if it is known implicitly.
\begin{align*}
\tp_0^v(x)&=\begin{cases}
1 & \text{$x=v$ and $1\geq 2\epsilon d(x)$,}\\
0 & \text{otherwise.}
\end{cases}
\\
\tp_t^v&=[T\tp_{t-1}]_{\epsilon}.
\end{align*}

The description of the algorithm \texttt{Nibble} and Lemma 3.1 in~\cite{spielman2004nearly} implies the following lemma.\footnote{There are many versions of the paper~\cite{spielman2004nearly} available; we refer to~\url{https://arxiv.org/abs/cs/0310051v9}.}

\begin{lemma}[\cite{spielman2004nearly}~]\label{le:Nibble}
For each $\phi\leq 1$, define the parameters
\begin{align*}
t_0 &= \frac{49\ln(me^4)}{\phi^2}\\
\text{and\ \ } \gamma &=\frac{5\phi}{392\ln(me^4)}.
\intertext{For each subset $S\subset V$ satisfying}
\vol(S) &\leq \frac{2}{3}\cdot \vol(V)\\
\text{and\ \ } 
\Phi(S) &\leq \frac{\phi^3}{19208\ln ^2(me^4)},
\end{align*}
there exists a subset $S^g\subseteq S$ with the following properties.
First, $\vol(S^g)\geq \vol(S)/2$.  Second, $S^g$ is partitioned into
$S^g = \bigcup_{b=1}^{\log m} S^g_b$ such that if a random walk is initiated
at any $v\in S^g_b$ with truncation parameter
$\epsilon=\frac{\phi}{56\ln(me^4)t_0 2^b}$, 
then there exists a number $t\in [1,t_0]$ 
and an index $j$ such that the following four conditions are met for the cut $C=\{\tpi_t^v(1), \ldots, \tpi_t^v(j))\}$.
\begin{itemize}
\item[(i)] $\Phi(C)\leq \phi$,
\item[(ii)] $\trho_t(\tpi_t(j))\geq \gamma/\vol(C)$,
\item[(iii)] $\vol(C\cap S)\geq (4/7)2^{b-1}$,
\item[(iv)] $\vol(C)\leq (5/6)\vol(V)$.
\end{itemize}
\end{lemma}

In subsequent discussion, with respect to a given parameter  $\phi\leq 1$, for any subset $S\subset V$ satisfying the condition of Lemma~\ref{le:Nibble}, we fix a subset $S^g\subseteq S$ and its decomposition $S^g = \bigcup_{b=1}^{\log m} S^g_b$ to be any choices satisfying  Lemma~\ref{le:Nibble}.

\subsection{Distributed Algorithm \label{subsect:dist-alg}}
Now we give our algorithm \texttt{Distributed Nibble}. 
To simplify things, we present it as a 
sequential algorithm, and prove in 
Lemma \ref{le:Nibblerunningtime} that 
it can be implemented efficiently in the $\CONGEST$ model.
For any permutation $\pi$, we use the notation $\pi(i..j)$ 
to denote the set $\{\pi(i), \pi(i+1), \ldots, \pi(j)\}$.

\begin{algorithm}[H]
\caption{\texttt{Distributed Nibble}}\label{al:disnibble}
\begin{algorithmic}
\item Input: $\phi$.
\For{parameter $b=1$ to $\lceil \log m \rceil$}

    \State Set parameters $t_0=49\ln(me^4)/\phi^2$, and $\epsilon_b=\frac{\phi}{56\ln(me^4)t_02^b}$, as in Lemma \ref{le:Nibble}.
    
    \parState{(1) Independently randomly sample $K=c \log m\cdot \frac{\vol(V)}{2^b}$ vertices $v_1,...,v_K$ proportional to their degrees, where $c$ is a large enough constant.}
    
    \parState{Initialize $\tp_0^{v_i}$.}
    
    \For{$t=1$ to $t_0$, for every $v_i$}

        \State(2) calculate $\tilde{p}_t^{v_i}=[T\tilde{p}_{t-1}^{v_i}]_{\epsilon_b}$

        \State Denote $j_{max}$ as the largest index such that $\tp_t^{v_i}(\tpi_t^{v_i}(j_{max}))>0$.

        \For{$x=0$ to $\log_{1+\phi} (5/6)\vol(V)$}

            \State(3) Set $j\leq j_{max}$ to be the largest index that $\vol(\tpi_t^{v_i}(1..j))\leq (1+\phi)^x$.

            \parState{(4) If $\Phi(\tpi_t^{v_i}(1..j))\leq 12\phi$, output the sparse cut $C=\tpi_t^{v_i}(1..j)$ and halt.}
        \EndFor

    \EndFor

\EndFor

\State Return $failed$.

\end{algorithmic}
\end{algorithm}

From Lemma~\ref{le:Nibble} we know that we can obtain a cut $C$ with some good properties if we start the truncated random walk at a vertex $v\in S^g_b$ with parameter $\epsilon_b$. 
Therefore, what we do in \texttt{Distributed Nibble} is to just sample sufficiently many vertices as the starting points of random walks so that with sufficiently high probability at least one them is in the set $S^g_b$. 
The danger here is that calculating all these random walk distributions simultaneously
may be infeasible if any part of the graph becomes \emph{too congested}.

In this section we analyze the behavior of \texttt{Distributed Nibble} (as a sequential algorithm)
and prove that it operates correctly.
In Section~\ref{subsect-impl} we
argue that \texttt{Distributed Nibble}
can be implemented efficiently in the $\CONGEST$ model, 
in $\poly(\log m, 1/\phi)$ time.\medskip

Roughly speaking, Lemma~\ref{le:approximation} shows that if the sets $\pi(1..j)$ and $\pi(1..j')$ have similar volume, 
then the cuts resulting from these two sets have similar sparsity. This justifies lines (3) and (4) of \texttt{Distributed Nibble} and allows us to examine a small number of prefixes of the permutation $\tpi_t^{v_i}$.

\begin{lemma}\label{le:approximation}
Let $\pi$ be any permutation,  and let $\phi\leq 1/12$. 
If, for some index $j$, 
$\Phi(\pi(1..j))\leq \phi$ and $\vol(\pi(1..j))\leq (5/6)\vol(V)$, 
then $\Phi(\pi(1..j'))\leq 12\phi$ for all indices $j'>j$ such that 
\[
\vol(\pi(1..j'))\leq (1+\phi)\vol(\pi(1..j)).
\]
\end{lemma}
\begin{proof}
Let $x=\vol(\pi(1..j))$ and $y=\vol(\pi(1..j'))$.
Recall that $2m = 2|E| = \vol(V)$, and so $x \leq (5/6)\vol(V) = (5/6)2m$.
We have $x\leq y\leq (1+\phi)x$. 
Since $x\leq (5/6)2m$ and $\phi\leq 1/12$, we have $\phi x \leq x/12 \leq (2m-x)/2$. Therefore,
\[
2m-y\geq 2m-x-\phi x \geq (2m-x)/2.
\]
We calculate an upper bound of $\Phi(\pi(1..j'))$ as follows.
\[
\Phi(\pi(1..j'))=\frac{\partial(\pi(1..j'))}{\min(y,2m-y)}\leq
\frac{\partial(\pi(1..j))+\sum_{i=j+1}^{j'}d(\pi(i))}{\min(x,(2m-x)/2 )}\leq \frac{\partial(\pi(1..j))+\phi x}{\min(x,(2m-x))/2}\leq 12\phi.
\]
We explain the details of the derivation.
The first inequality is due to $x\leq y$ and $(2m-x)/2  \leq 2m-y$, which follow from the above discussion.
The second inequality is due to the fact that $\sum_{i=j+1}^{j'}d(\pi(i)) = \vol(\pi(1..j')) - \vol(\pi(1..j)) \leq \phi\cdot \vol(\pi(1..j)) = \phi x$. For the third inequality, note that $\frac{\partial(\pi(1..j))}{\min(x,(2m-x))} \leq \phi$ and $\frac{\phi x}{\min(x,(2m-x))}\leq 5\phi$, since $x\leq (5/6)2m$.
\end{proof}

\begin{lemma}\label{le:sg}
Let $S\subset V$ be any subset satisfying
\[
\vol(S)\leq (2/3)\vol(V) ~~~\text{and}~~~ \Phi(S)\leq \frac{\phi^3}{19208\ln ^2(me^4)}.
\]
Then there exists a number $b$ such that $\vol(S^g_b)\geq 2^b/32$.
\end{lemma}
\begin{proof} Denote $x=\vol(S)$. From Condition~(iii) of Lemma~\ref{le:Nibble} we deduce that if $S^g_b\neq \emptyset$, then there exists a set of vertices $C$ such that $\vol(S)\geq \vol(C\cap S)\geq (4/7)2^{b-1}$. Thus, for all $b$ such that $b\geq \lceil \log x \rceil +2$, we must have $S^g_b=\emptyset$. If the statement of this lemma is false, i.e., $\vol(S^g_b)<2^b/32$ for all $b$, then
\[
\vol(S^g)\leq \sum_{b=1}^{\lceil \log x \rceil +1}\frac{2^b}{32} < \frac{2^{\lceil \log x \rceil + 2}}{32} < x/4,
\]
which contradicts the 
requirement $\vol(S^g)\geq \vol(S)/2$ specified in Lemma~\ref{le:Nibble}.
\end{proof}

\begin{lemma}[Correctness]\label{le:correctness}
For any $\phi\leq 1/12$, if there exists a subset $S\subset V$ satisfying
\[
\vol(S)\leq (2/3)\vol(V) ~~~\text{and}~~~ \Phi(S)\leq \frac{\phi^3}{19208\ln ^2(me^4)},
\]
then \texttt{Distributed Nibble} outputs a set of vertices $C$
such that $\Phi(C)\leq 12\phi$ with probability at least $1-1/\poly(m)$.
\end{lemma}
\begin{proof} From Lemma~\ref{le:sg} we know there exists a number $b$ such that $\vol(S^g_b)\geq 2^b/32$. Since we sample $v_i$ proportional to the degree distribution,
\[
\Pr[v_i\in S^g_b]=\frac{\vol(S^g_b)}{\vol(V)}\geq \frac{2^b}{32\cdot \vol(V)}.
\]
Since we sample $K=c \log m\cdot \frac{\vol(V)}{2^b}$ number of vertices, \[
\Pr[\exists i \text{ s.t. } v_i\in S^g_b]
\geq 1-\left(1- \frac{2^b}{32\vol(V)}\right)^{c\log m \frac{\vol(V)}{2^b}}
\geq 1-m^{-\Omega(c)}.
\]

Now we focus on the truncated random walk starting at this vertex $v_i \in S^g_b$. We fix two numbers $t \in[1,t_0]$ and $j$ such that the four conditions in Lemma~\ref{le:Nibble} are satisfied. 
In particular,  Condition (i) and Condition~(iv) in Lemma~\ref{le:Nibble} say that 
\begin{align*}
    \vol(\tpi_t^{v_i}(1..j)) &\leq (5/6)\vol(V),\\
    \Phi(\tpi_t^{v_i}(1..j))&\leq \phi.
\end{align*}
Therefore, we are able to apply Lemma~\ref{le:approximation}, and so we have $\Phi(\tpi_t^{v_i}(1..j'))\leq 12\phi$ for all indices $j'$ such that $\vol(\tpi_t^{v_i}(1..j)) \leq \vol(\tpi_t^{v_i}(1..j'))\leq (1+\phi)\vol(\tpi_t^{v_i}(1..j))$.

In \texttt{Distributed Nibble}, we search for a cut with target volume $(1+\phi)^x$, for all possible integers $x$. 
Note that Condition (ii) in Lemma~\ref{le:Nibble} implies $j\leq j_{max}$.
Therefore, in Step~(3) of \texttt{Distributed Nibble}, at least one index $j^\star$ picked by the algorithm satisfies
\begin{align*}
    &\vol(\tpi_t^{v_i}(1..j)) \leq \vol(\tpi_t^{v_i}(1..j^\star))\leq (1+\phi)\vol(\tpi_t^{v_i}(1..j)). 
\end{align*}
By Lemma~\ref{le:approximation}, 
the cut $C= \tpi_t^{v_i}(\{1,...,j^\star\})$ associated with this index $j^\star$ found in Step~(4)  meets the requirement $\Phi(C)\leq 12\phi$ of the lemma.
\end{proof}

\subsection{Implementation \label{subsect-impl}}

We show how to implement \texttt{Distributed Nibble} in the $\CONGEST$ model. 
The goal of this section is to prove Lemma~\ref{le:Nibblerunningtime}.
Note that Lemma~\ref{le:lowconductance} is a consequence of Lemmas~\ref{le:correctness} and~\ref{le:Nibblerunningtime}.

\begin{lemma}\label{le:Nibblerunningtime}
\texttt{Distributed Nibble} can be implemented in the $\CONGEST$ model using $O(D + \log^9 m/\phi^{10})$ 
rounds, with success probability $1 - 1/\poly(m)$, where $D$ is the diameter of graph. If \texttt{Distributed Nibble} outputs a set $C$ successfully, then each vertex knows whether or not it belongs to $C$.
\end{lemma}

To prove this lemma, we shall analyze \texttt{Distributed Nibble} step by step.

\begin{lemma}[Step (1)]\label{le:step1}
The samples for every level $b$ (from 1 to $\lceil \log m \rceil$) can be generated in $O(D + \log m)$ time.
\end{lemma}

\begin{proof}
We build a BFS tree rooted at an arbitrary vertex $x$. For each vertex $v$, define $s(v)$ as the sum of $d(u)$ for each $u$ in the subtree rooted at $v$. In $O(D)$ rounds we can let each vertex $v$ learn the number $s(v)$ by a bottom-up traversal of the BFS tree. 

In the beginning, for each  $b=1,\ldots,\lceil \log m \rceil$,  we generate $K_b = c \log m \frac{\vol(V)}{2^b}$ number of $b$-tokens at the root $x$. 
Let $L=\Theta(D)$ be the number of layers in the BFS tree. For $i = 1, \ldots, L$, the vertices of layer $i$ do the following. When a $b$-token arrives at $v$, 
the token disappears at $v$ with probability $d(v)/s(v)$
and $v$ includes itself in the $b$th sample;
otherwise, $v$ sends the token to a child $u$ 
with probability $\frac{s(u)}{s(v)-d(v)}$. 
Note that $v$ only needs to tell each child $u$ 
how many $b$-tokens $u$ gets.  Thus, for each $b$, the process of choosing $K_b = c \log m \frac{\vol(V)}{2^b}$ vertices from the degree distribution can be done in 
$L$ rounds. By pipelining, we can do this for all $b$ in $O(D+\log m)$ rounds.

This method has the virtue of selecting exactly $K_b$ vertices in the $b$th sample.  We can also select $K_b$ vertices in expectation, in just $O(D)$ time, simply by computing $\vol(V)$ with a BFS tree, disseminating it to all vertices,
and letting each $v$ join the sample independently with probability $K_b \deg(v)/\vol(V)$.
\end{proof}

It is not obvious why Step (2) of \texttt{Distributed Nibble} should be efficiently implementable in the $\CONGEST$ model.  
Before analyzing it, we give some helpful lemmas about lazy random walks.

\begin{lemma}[\cite{spielman2004nearly}~]\label{le:rho}
For all $u,v,$ and $t$, $\rho_t^v(u)=\rho_t^u(v)$.
\end{lemma}
\begin{proof}
This lemma was observed in~\cite{spielman2004nearly} without proof.
For the sake of completeness, we provide a short proof here.
A sequence of vertices $W=(x_0, x_1, \ldots, x_t)$ is called a {\em walk} of length $t$ if $x_{i+1} \in N(x_i) \cup \{x_i\}$ for each $i\in [0,t)$. 
We write $\Prob[W]$ to be the probability that the first $t$ steps of a lazy random walk starting at $x_0$ tracks $W$.
Let $W^R = (x_t, x_{t-1}, \ldots, x_0)$ be the reversal of $W$.

Let $\mathcal{W}_t^{u,v}$ be the set of walks of length $t$ starting at $u$ and ending at $v$. 
It is clear that 
$\rho_t^u(v) = \sum_{W \in \mathcal{W}_t^{u,v}} \Prob[W] / d(v)$ and 
$\rho_t^v(u) = \sum_{W \in \mathcal{W}_t^{v,u}} \Prob[W] / d(u)$.
Since $\mathcal{W}_t^{v,u} = \{ W^R \mid W \in \mathcal{W}_t^{u,v}\}$, to prove the lemma it suffices to show that
$\Prob[W] / d(v) = \Prob[W^R] / d(u)$ for each $W \in \mathcal{W}_t^{u,v}$.

Fix any $W \in \mathcal{W}_t^{u,v}$
and let $W_{\ast}=(y_0, \ldots, y_s)$ be the subsequence of $W$ 
resulting from splicing out immediate repetitions in $W$.
It is clear that $\Prob[W] = 2^{-t} \cdot \prod_{i=0}^{s-1} 1/d(y_i)$, 
and so 
\[
\frac{\Prob[W]}{d(v)} 
= \frac{\Prob[W]}{d(y_s)} 
= 2^{-t} \cdot \prod_{i=0}^{s} \frac{1}{d(y_i)}
= \frac{\Prob[W^R]}{d(y_0)}
= \frac{\Prob[W]}{d(u)}.\qedhere
\]
\end{proof}

\begin{lemma}\label{le:haha}
Fix the parameter $b$ (which influences $\epsilon_b$ and hence
the truncation operation of the random walk)
and define 
\[
Z_t(u)=\{v_i \mid \mbox{$v_i$ is in the $b$th sample and } \tp_{t-1}^{v_i}(u)>0\}.
\]
For every vertex $u$ and every $t$, with high probability,
$|Z_t(u)|\leq O(\log^3 m/\phi^3)$.
\end{lemma}
\begin{proof} 
Define $S=\{v\in V \mid \tp_{t-1}^v(u)>0\}$. 
By definition $Z_t(u) = S \cap \{v_1, \ldots, v_{K_b}\}$.
For each $v\in S$, we have 
$p_{t-1}^v(u)\geq \tp_{t-1}^v(u)\geq 2\epsilon_b d(u)$.
Recall that $p_{t-1}$ is the probability distribution obtained after $t-1$ 
steps of the lazy random walk without truncation. By Lemma~\ref{le:rho}, 
\[
p_{t-1}^u(v)= (p_{t-1}^v(u)/d(u)) d(v)\geq 2\epsilon_b\cdot d(v).
\]
Therefore,
$2\epsilon_b\cdot \vol(S) \leq \sum_{v \in S} p_{t-1}^u(v) \leq 1$, and so
$\vol(S)\leq \frac{1}{2\epsilon_b}$, which implies
\[
\Pr[v_i\in S]\leq \frac{1}{2\epsilon_b\cdot \vol(V)}.
\]

Recall that  $t_0=\frac{49\ln(me^4)}{\phi^2}$ and $\epsilon_b=\frac{\phi}{56\ln(me^4)t_02^b}$.
Rewrite the number $K_b = c \log m \frac{\vol(V)}{2^b}$   as 
$K_b =\Theta(\epsilon_b\cdot \vol(V)\cdot \log^3m/\phi^3)$.
Since each of $v_1, \ldots, v_{K_b}$ is chosen independently, 
using a Chernoff bound we conclude 
that there exists a constant $c' > 0$ depending on $c$ such that 
\[
\Pr[|Z_t(u)|> c'\log^3 m/\phi^3]\leq \exp(-\Omega(\log ^3 m/\phi^3)).\qedhere
\]
\end{proof}

\begin{lemma}[Step (2)]\label{le:step2}
Fix the parameter $b$.
Suppose each vertex $v$ knows  $\tp_{t-1}^{v_i}(v)$, 
for all $v_i$ in the $b$th sample.  
Then with high probability, 
each vertex $v$ can calculate $\tp_t^{v_i}(v)$, 
for all $v_i$, within $O(\log^3 m/\phi^3)$ rounds.
\end{lemma}
\begin{proof}
The normal way to calculate $[T\tp_{t-1}(u)]_{\epsilon_b}$ 
is as follows. For each $v_i$, each vertex $v$ broadcasts the number $\frac{\tp_{t-1}^{v_i}(v)}{2d(v)}$ to all its neighbors, 
and then $v$ collects messages from neighbors. The vertex $v$ can calculate  $\tp_t^{v_i}(v)$ locally by adding $\tp_{t-1}^{v_i}(v)/2$ and all numbers received 
from its neighbors, then applying the truncation operation.
Note that a straightforward analysis of this protocol leads to a terrible round complexity, since we have to do this for each $v_i$.

One crucial observation is that a vertex $v$ does not need to care about those $v_i$ with $\tp_{t-1}^{v_i}(v)=0$ at time $t$. We modify this protocol a little bit in such a way that we never send a number if it is $0$. Define $Z_t(u)=\{v_i \mid \tp_{t-1}^{v_i}(u)>0\}$ as in Lemma~\ref{le:haha}, and so each vertex $v$ only needs to spend $|Z_t(v)|$ rounds to simulate the time step $t$ of the lazy random walk. 
By Lemma~\ref{le:haha} and the discussion above, we have proved that Step~(2) can be executed in $O(\log^3 m/\phi^3)$ time, for every $v_i$ and any specific $t$.
\end{proof}

\begin{lemma}[Steps (3,4)]\label{le:step34}
Fix parameters $b$, $t$ and $x$. Steps (3) and (4) can be implemented in $O(\log^6 m/\phi^7)$ rounds for all $v_i$ in the $b$th sample. 
For any sparse cut $C$ found in Step (4), every vertex in $C$ knows that it belongs to $C$.
\end{lemma}
\begin{proof}
Now we focus on the random walk starting at $v_i$. Let $U=\{u \mid \exists t'\leq t, \tp_{t'}^{v_i}(u)>0\}$. $U$ must be a connected vertex set. Obviously all $\tpi_t^{v_i}(j)$ for $j\leq j_{max}$ are in $U$. 
We build a BFS tree of $U$ rooted at $v_i$, which has $t+1$ levels. We will execute Step~(3) and Step~(4) by sending requests from the root to all vertices in $U$, collecting information from $U$ to the root, and making a decision locally at the root. Recall that each $v_i$ has its own BFS tree, and in general a vertex $u$ belongs to multiple BFS trees for different $v_i$. 
Luckily, each vertex $u$ only belongs to the BFS tree of those $v_i\in \bigcup_{1 \leq t \leq t_0} Z_t(u)$, so with only a $t_0\cdot \max_{u, t} |Z_t(u)| = O(\log^4 m/\phi^5)$ overhead of running time, we can do Step~(3) and Step~(4) for all $v_i$ in parallel. 

To find each index $j$ specified in Step~(3), we can do a ``random binary search'' on vertices in $U$. 
Let $\tpi = \tpi_t^{v_i}$ and $\tilde{\rho} = \tilde{\rho}_t^{v_i}$ be with respect to $\tp_t^{v_i}$. 
Note that by our choice of $U$ we can assume $U$ is a prefix set of $\tpi$. We maintain two indices $L$ and $R$ that control the search space. Initially, $L \gets 1$ and $R \gets |U|$. In each iteration, we randomly pick one vertex $\tpi(j)$ among $\tpi(L..R)$ and calculate $\vol(\tpi(1,\cdots,j))$ by 
broadcasting $\tilde{\rho}(\tpi(j))$ to all vertices in $U$ and propagating information up the BFS tree. If $\vol(\tpi(1,\cdots,j))\leq (1+\phi)^x$, we update $L \gets j$; otherwise we let $R=j-1$. In each iteration, with probability $1/2$ we sample $j$ in the middle half of $[L,R]$ and the size of search space 
$[L,R]$ shrinks by a factor of at least $3/4$. Therefore, w.h.p., 
after $O(\log m)$ iterations, we will have isolated $L = R = j$. 
Each iteration can be done in $O(t)=O(t_0)$ rounds. 
Due to the congestion overhead, Step~(3) can be implemented in $O(\log m \cdot t_0 \cdot \log^4 m/\phi^5) 
= O(\log^6 m/\phi^7)$ rounds.

Step~(4) can be done by simply collecting information about $\partial(\tpi_t^{v_i}(1..j))$ and  $\vol(\tpi_t^{v_i}(1..j))$; its round complexity
is of a lower order than that of Step~(3).
If the root $v_i$ finds a cut $C$  with $\Phi(C)\leq 12\phi$, it broadcasts $\tilde{\rho}(\tpi(j))$ to all vertices in $U$ to let the vertices in $C$ know that they are in $C$. 
Note that for each vertex $u$ in $U$, 
it can infer whether it is in $C$ by comparing 
$\tilde{\rho}(u)$ and $\tilde{\rho}(\tpi(j))$. 
\end{proof}

\begin{proof}[Proof of Lemma \ref{le:Nibblerunningtime}]
Combining Lemmas~\ref{le:step1}, \ref{le:step2},
and \ref{le:step34},
the running time in Step~(1) is $O(D+\log m)$, 
Step~(2) is $O(\log^5m/\phi^5)$,
and Steps~(3) and (4) are $O(\log^9m/\phi^{10})$. 
The dominating term $O(\log^9m/\phi^{10})$ 
comes from enumerating $\log m\cdot t_0\cdot \log m/\phi = \Theta(\log^3 m/\phi^3)$ combinations of $(b,t,x)$, spending 
$O(\log^6 m/\phi^7)$ rounds for each combination.

Whenever a vertex $v_i$ finds a sparse cut $C$, it broadcasts a message to the entire graph saying that it has found a cut, and this takes $O(D)$ rounds.
If multiple cuts are found by different vertices, we 
can select exactly one cut, breaking ties arbitrarily.  
A more opportunistic
version of the algorithm could also take a 
maximal independent set of compatible cuts.
\end{proof}

\section{Triangle Enumeration}\label{se:triangle}

We use the routing algorithm from~\cite{GhaffariKS17,GhaffariL2018}.
Theorem~\ref{th:mixing} was first stated in~\cite[Theorem~1.2]{GhaffariKS17} with round complexity  $\mix(G)\cdot 2^{O(\sqrt{\log n\log\log n})}$; this was recently improved to $\mix(G)\cdot 2^{O(\sqrt{\log n})}$ in~\cite{GhaffariL2018}.

\begin{theorem}[\cite{GhaffariKS17,GhaffariL2018}~]\label{th:mixing}
Consider a graph $G=(V,E)$ and a set of point-to-point routing requests, each given by the $\ID$s of the corresponding source-destination pair. If each vertex $v$ is the source and the destination of at most $\deg(v)\cdot 2^{O(\sqrt{\log n})}$ messages, there is a randomized distributed algorithm that delivers all messages in $\mix(G)\cdot 2^{O(\sqrt{\log n})}$ rounds, w.h.p., in the $\CONGEST$ model.
\end{theorem}

\begin{rem}
The claim of Theorem~\ref{th:mixing} appears to 
be unproven for arbitrary $\ID$-assignments (Hsin-Hao Su, personal communication, 2018), but is true
for well-behaved $\ID$-assignments, 
which we illustrate can be computed efficiently in $\CONGEST$.  In \cite{GhaffariKS17,GhaffariL2018} each
vertex $v\in V$ simulates $\deg(v)$ virtual vertices in a random graph $G_0$ which is negligibly close to one drawn
from the \Erdos-\Renyi{} distribution $\mathcal{G}(2m,p)$ for some $p$.  Presumably the IDs
of $v$'s virtual vertices are $(\ID(v),1),\ldots,(\ID(v),\deg(v))$.
It is proven~\cite{GhaffariKS17,GhaffariL2018} 
that effecting a set of routing requests in $G_0$ takes $2^{O(\sqrt{\log n})}$ time in $G_0$;
however, to translate a routing request $\ID(x)\leadsto \ID(y)$ in $G$ to $G_0$,
it seems necessary to map it (probabilistically) to $(\ID(x),i)\leadsto (\ID(y),j)$,
where $i,j$ are chosen uniformly at random from $[1,\deg(x)]$ and $[1,\deg(y)]$, respectively.
(This is important for the global congestion guarantee that 
$y$'s virtual nodes receive roughly equal numbers of messages from all sources.)
This seems to require that $x$ know how to compute $\deg(y)$ or an approximation thereof
based on $\ID(y)$.  Arbitrary $\ID$-assignments obviously do not betray this information.
\end{rem}

\begin{lemma}\label{le:knowaboutlabel}
In $O(D+\log n)$ time we can compute an $\ID$ assignment
$\ID : V\rightarrow\{1,\ldots,|V|\}$
and other information such that
$\ID(u)<\ID(v)$ implies $\lfloor\log\deg(u)\rfloor \le \lfloor\log\deg(v)\rfloor$,
and any vertex $u$ can locally compute $\lfloor\log\deg(v)\rfloor$ for any $v$.
\end{lemma}

\begin{proof}
Build a BFS tree from an arbitrary vertex $x$ in $O(D)$ time.
In a bottom-up fashion, each vertex in the BFS tree calculates the \emph{number}
of vertices $v$ in its subtree having $\lfloor\log\deg(v)\rfloor = i$, for $i=0,\ldots,\log n$.
This takes $O(D + \log n)$ time by pipelining.
At this point the root $x$ has the counts $n_0,\ldots,n_{\log n}$ for each degree class,
where $n = \sum_i n_i$.  It partitions up the $\ID$-space so that all vertices 
in class-0 get $\ID$s from $[1,n_0]$, class-1 from $[n_0+1,n_0+n_1]$, and so on.
The root broadcasts the numbers $n_0,\ldots,n_{\log n}$, 
and disseminates the $\ID$s to all nodes according to their degrees.
(In particular, the root gives each child $\log n$ intervals of the $\ID$-space,
which they further subdivide, sending $\log n$ intervals to the grandchildren, etc.)
With pipelining this takes another $O(D+\log n)$ time.
Clearly knowing $n_0,\ldots,n_{\log n}$ and $\ID(v)$ suffice to calculate $\lfloor\log\deg(v)\rfloor$.
\end{proof}

Lemma~\ref{le:knowaboutlabel} gives us a good $\ID$-assignment to apply
Theorem~\ref{th:mixing}.  It is also useful in our triangle enumeration application.
Roughly speaking, vertices with larger degrees also have more bandwidth in the $\CONGEST$ model,
and therefore 
should be responsible for learning about larger subgraphs and enumerating more triangles.

Before we present our triangle enumeration algorithm for general graphs,
we address the important special case of finding triangles with at least one
edge in a component of high conductance.


\newcommand{\Gin}{G_{\operatorname{in}}}
\newcommand{\Vin}{V_{\operatorname{in}}}
\newcommand{\Ein}{E_{\operatorname{in}}}
\newcommand{\Eout}{E_{\operatorname{out}}}
\newcommand{\degin}{\deg_{\operatorname{in}}}
\newcommand{\degout}{\deg_{\operatorname{out}}}
\newcommand{\mbar}{\bar{m}}
\newcommand{\nbar}{\bar{n}}

\subsection{Triangle Enumeration in High Conductance Graphs}\label{subsect:lem}

Recall that our graph decomposition routine returns a tripartition $E_m \cup E_s \cup E_r$. Triangles that intersect $E_s$ will be enumerated separately.  The purpose
of this section is to provide a routine to enumerate triangles that intersect $E_m$, i.e.,
they are completely contained in $E_m$ or have one edge in $E_m$ and two in $E_r$.
Whereas each component of $E_m$ has low mixing time, we can say nothing about the mixing time of $E_m$ plus all incident $E_r$ edges.

\paragraph{Definitions.}
The input is a subgraph $\Gin = (\Vin,\Ein)$ with low mixing time,
together with some edges $\Eout$ joining vertices in $\Vin$ to $V-\Vin$.
Let $\degin(v)$ and $\degout(v)$ be the
number of $\Ein$ and $\Eout$ edges incident to $v$.
In this section $n = |\Vin|$ and $m = |\Ein|$.

\begin{theorem}\label{le:mainlemma}
Suppose that $\Gin$ and $\Eout$ meet the following conditions:
\begin{itemize}
    \item[(i)] For each $v\in \Vin$, $\degin(v) \ge \degout(v)$.
    \item[(ii)] $\mix(\Gin) = n^{o(1)}$.
\end{itemize}
In the $\CONGEST$ model, all triangles in $\Ein \cup \Eout$ can
be counted and enumerated, w.h.p., in $O(n^{1/3+o(1)})$ rounds.
\end{theorem}

Note that Theorem~\ref{le:mainlemma}
applies to the class of graphs with $n^{o(1)}$ mixing time,
letting $\Eout=\emptyset$.
We first describe the algorithm behind Theorem~\ref{le:mainlemma} and then
analyze it in Lemmas~\ref{le:subgraphedge}--\ref{le:sendmessage}.

\paragraph{The Easy Case.} We first check whether any vertex
$v^\star \in V(\Ein \cup \Eout)$ has
$\degin(v^\star)+\degout(v^\star) \ge m/(20n^{1/3}\log n) = \zeta$.
If so, we apply Theorem~\ref{th:mixing} to the subgraph $\Gin^+$
induced by $\Vin\cup\{v^\star\}$ and have every vertex
$u\in \Vin$ transmit to $v^\star$ all its incident edges in $\Ein\cup\Eout$.\footnote{Notice that when $v^\star \in V(\Eout) - \Vin$,
$\mix(\Gin) = n^{o(1)}$ implies that $\mix(\Gin^+)=n^{o(1)}$ as well.}
Condition (i) of Theorem~\ref{le:mainlemma} implies that
$|\Ein \cup \Eout| \le 3m$,
so the total volume of messages entering $v^\star$ is $O(m)$.
Therefore the routing takes
$O(\mix(\Gin^+)\cdot 2^{O(\sqrt{\log n})}\cdot m/\zeta)
= n^{1/3 + o(1)}$,
and thereafter, $v^\star$ can report all triangles in $\Ein\cup\Eout$.
In the analysis of the following steps, we may assume that
the maximum degree in the graph induced
by $\Ein\cup \Eout$ is at most $m/(20n^{1/3}\log n)$.

\paragraph{Vertex Classes.}
Let $\delta = 2^{\lfloor \log(2m/n)\rfloor}$ be the average degree in $\Gin$,
rounded down to the nearest power of 2.
Write $\degin(v) = k_v \cdot \delta$, and call
$v$ a \emph{class-0 vertex} if $k_v \in [0,1/2)$ and a
\emph{class-$i$ vertex}
if $k_v \in [2^{i-2},2^{i-1})$.  We use the fact that
\[
\sum_{v\in \Vin \,:\, k_v \ge 1/2} 2k_v \ge n.
\]
By applying Lemma~\ref{le:knowaboutlabel} to reassign IDs,
we may assume that the $\ID$-space of $\Vin$ is $\{1,\ldots,|\Vin|\}$
and that any vertex can compute the class of $v$,
given $\ID(v)$.

\paragraph{Randomized Partition.}
Our algorithm is a randomized adaptation of the $\CLIQUE$ algorithm of \cite{DolevLP12}.
We partition the set $V(\Ein\cup \Eout)$ into $V_1\cup \cdots \cup V_{n^{1/3}}$
locally, without communication.  Each vertex $v\in V(\Ein\cup\Eout)$
selects an integer $r_v \in [1,n^{1/3}]$ uniformly at random,
joins $V_{r_v}$, and transmits `$r_v$' to its immediate neighbors.
We allocate the (less than) $n$ triads
\[
\mathscr{T} = \left\{(j_1,j_2,j_3) \mid 1 \le j_1 \le j_2 \le j_3 \le n^{1/3}\right\}
\]
to the vertices in $\Vin$ in the following way.
Enumerate the vertices in increasing order of ID.
If $v$ is class-$0$, then skip $v$.
If $v$ is class-$i$, $i\ge 1$, then $k_v < 2^{i-1}/\delta$.
Allocate to $v$ the next $2^i/\delta \ge 2k_v$ triads from $\mathscr{T}$,
and stop whenever all triads are allocated.
By Lemma~\ref{le:knowaboutlabel},
every vertex can calculate the class of any other, and can therefore
perform this allocation locally, without communication.
A vertex $v \in V$ that is assigned a triad $(j_1,j_2,j_3)$ is responsible of learning the set of all edges $E(V_{j_1},V_{j_2})\cup E(V_{j_2},V_{j_3})\cup E(V_{j_1},V_{j_3})$ and reporting/counting those triangles $(x_1,x_2,x_3)$
with $x_k \in V_{j_k}$.\footnote{In the Triangle Counting application,
it is important that $v$ not count every triangle it is aware of.
For example, if $v$ is assigned $(j,j,j')$, $v$ knows about triangles
in the subgraph induced by $V_{j}$ but should not count them.}

\paragraph{Transmitting Edges.}
Every vertex knows the $\ID$s of its neighbors and which part
of the vertex partition they are in.  For each $v\in \Vin$,
each incident edge $(v,u) \in \Ein\cup \Eout$,
and each index $r^* \in [1,n^{1/3}]$,
$v$ transmits the message ``$(v,u),r_v,r_u$'' to
the unique vertex $x$ handling the triad on $\{r_u,r_v,r^*\}$.
Observe that the total message volume is exactly $\Theta(mn^{1/3})$.

We analyze the behavior of this algorithm in the $\CONGEST$ model,
where the last step is implemented by applying Theorem~\ref{th:mixing}
to $\Gin$.  Recall from Condition (i)
of Theorem~\ref{le:mainlemma} that the
number of edges in the graph we consider,
$\mbar = |\Ein\cup\Eout|$, is in the range $[m,3m]$.

\begin{lemma}\label{le:subgraphedge}
Consider a graph with $\mbar$ edges and $\nbar$ vertices.
We generate a subset $S$ by
letting each vertex join $S$ independently
with probability $p$.
Suppose that the maximum degree
is $\Delta \leq \mbar p/20\log \nbar$ and
$p^2\mbar \geq 400\log^2 \nbar$.
Then, w.h.p., the number of edges in the
subgraph induced by $S$ is at most $6p^2\mbar$.
\end{lemma}
\begin{proof}
For an edge $e_i$, define $x_i = 1$ if both two endpoints of edge $e_i$ join $S$,
otherwise $x_i = 0$.
Then $X=\sum_{i=1}^{\mbar} x_i$ is the number of edges in the subgraph induced by $S$. We have $\Exp[X]=p^2\mbar$, and by Markov's inequality,
\[
\Pr[X\geq 6\Exp[X]]=\Pr[X^c \geq (6\Exp[X])^c]\leq \frac{1}{6^c}\frac{\Exp[X^c]}{p^{2c}\mbar^c},
\]
where $c=5\log \nbar$ is a parameter.
\begin{align*}
    \Exp[X^c]&=\sum_{i_1,\ldots,i_c\in[1,\mbar]} \Exp\left[\prod_{j=1}^c x_{i_j}\right]\\
    & =\sum_{k=2}^{2c} f_k\cdot p^k,
\end{align*}
where $f_k$ is the number of choices $\{i_1,\ldots,i_c\in[1,\mbar]\}$ such that the number of distinct endpoints in the edge set $e_{i_1},\ldots,e_{i_c}$ is $k$.

For any choice of $(i_1,\ldots,i_c\in[1,\mbar])$, we project it to a vector
$\langle k_1,\ldots,k_c\rangle\in \{0,1,2\}^c$,
where $k_j$ indicates the number of endpoints of $e_{i_j}$ that overlap with the endpoints of the edges $e_{i_1},\ldots,e_{i_{j-1}}$. Note that $2c-\sum k_j$ is the number of distinct endpoints in the edge set $\{e_{i_1},\ldots, e_{i_c}\}$. We fix a vector $\langle k_1,\ldots,k_c\rangle$ and count how many choices of $(i_1,\ldots,i_c)$ project to this vector.

Suppose that the edges  $e_{i_1},\ldots,e_{i_{j-1}}$ are fixed. We bound the number of choices of $e_{i_j}$ as follows. If $k_j = 0$, the number of choices is clearly at most $m$.
If $k_j=1$, the number of choices is at most $(2c)(p\mbar/20\log \nbar)$,
since one of its endpoints (which overlaps with the endpoints of the edges
$e_{i_1},\ldots,e_{i_{j-1}}$) has at most $2c$ choices, and the other endpoint (which does not overlap with the endpoints of the edges $e_{i_1},\ldots,e_{i_{j-1}}$)
has at most $\Delta \leq \mbar p/20\log \nbar$ choices. If $k_j=2$, the number of choices is at most $(2c)^2$.

Based on the above calculation, we upper bound $f_k$ as follows.
In the calculation, $x$ is the number of indices $j$ such that $k_j=1$, and $y$ is the number of indices $j$ that $k_j=2$.
Note that ${c\choose x}{c-x\choose y}$ is the number of distinct vectors $\langle k_1,\ldots,k_c\rangle$ realizing the given parameters $c$, $x$, and $y$.
\begin{align*}
    f_k&\leq \sum_{x+y\leq c, \; 2c-x-2y=k} \mbar^{c-x-y}{c\choose x}{c-x\choose y}\left(\frac{2cp\mbar}{20\log \nbar}\right)^x(4c^2)^y\\
    &\leq \sum_{x+y\leq c, \; 2c-x-2y=k} \mbar^c3^c\left(\frac{2cp}{20\log \nbar}\right)^x\left(\frac{4c^2}{\mbar}\right)^y\\
    &\leq \sum_{x+y\leq c, \; 2c-x-2y=k} (3\mbar)^c \left(\frac{2cp}{20\log \nbar}\right)^{x+2y}\\
    &\leq c(3\mbar)^c\left(\frac{2cp}{20\log \nbar}\right)^{2c-k}.
\end{align*}
The third inequality is due to the fact
$p^2\mbar \geq 400\log^2 \nbar$, which implies $(2cp/20\log \nbar)^2\geq (4c^2/\mbar)$.
Using the fact that $\frac{2c}{20\log \nbar}\leq 1/2$, we upper bound $\Exp[X^c]$ as follows.
\begin{align*}
    \Exp[X^c]&\leq \sum_{k=2}^{2c} f_k\cdot p^k\\
    &=c(3\mbar)^cp^{2c} \sum_{k=2}^{2c}\left(\frac{2c}{20\log \nbar}\right)^{2c-k}\\
    &< 2c(3\mbar)^cp^{2c}.
\end{align*}
Therefore,
\[
\Pr[X\geq 6\Exp[X]]\leq \frac{1}{6^c}\frac{\Exp[X^c]}{p^{2c}\mbar^c}\leq \frac{2c3^c}{6^c}\leq \frac{10\log \nbar}{\nbar^5}.
\]
Note that the probability can be amplified to
$\nbar^{-t}$ for any constant $t$ by setting
$c = t \log \nbar$ and using different
constants in the statement of the lemma.
\end{proof}

\begin{lemma}\label{le:edgenumber}
For all $j_1,j_2 \in [1,n^{1/3}]$,
with high probability,
$|E(V_{j_1},V_{j_2})|\leq 24m/n^{2/3}$.
\end{lemma}
\begin{proof}
For the case $j_1=j_2$,
we apply Lemma~\ref{le:subgraphedge}
to the subgraph on $\Ein\cup\Eout$
having $\mbar$ edges and $\nbar$ vertices,
with sampling probability $p=n^{-1/3}$.
One may verify
that the maximum degree is at most
$mp/(20\log n) < \mbar p/(20\log \nbar)$
and that
$p^2\mbar \geq n^{1/3} \geq 400\log^2 \nbar$.
By Lemma~\ref{le:subgraphedge}, we conclude that $\Pr[|E(V_{j_1})| > 6\mbar/n^{2/3}]\leq
\frac{10\log n}{n^5}$.

For the case $j_1 \neq j_2$, we can apply
the same analysis to the subgraph induced
by $V_{j_1}\cup V_{j_2}$ with $p=2n^{-1/3}$
and conclude that
$\Pr[|E(V_{j_1},V_{j_2})|>24\mbar/n^{2/3}]
\leq \frac{10\log n}{n^5}$.
By a union bound over all $n^{2/3}$ choices
of $j_1$ and $j_2$,
the stated upper bound holds everywhere, w.h.p.
\end{proof}

\begin{lemma}\label{le:receivemessage}
With high probability,
each vertex $v\in \Vin$
receives $O(\degin(v)\cdot n^{1/3})$ edges.
\end{lemma}
\begin{proof}
Consider any vertex $v \in \Vin$.
If $k_v<1/2$, then $v$ receives no message;
otherwise $v$ is responsible for between
$2k_v$ and $4k_v$ triads, and $v$ collects
the edge set $E(V_{j_1},V_{j_2})$
for at most $12k_v$ pairs of $V_{j_1}$, $V_{j_2}$.
By Lemma~\ref{le:edgenumber},
w.h.p., $|E(V_{j_1},V_{j_2})| = O(m/n^{2/3})$.
Hence $v$ receives
\[
O(m/n^{2/3}) \cdot 12k_v
=
O(\degin(v)\cdot n^{1/3})
\]
messages, with high probability.
\end{proof}

\begin{lemma}\label{le:sendmessage}
Each vertex $v\in \Vin$ sends $O(\degin(v)\cdot n^{1/3})$ edges with probability 1.
\end{lemma}
\begin{proof}
By Condition (i) of Theorem~\ref{le:mainlemma},
$v\in \Vin$ is responsible for
$\degin(v) + \degout(v) \le 2\degin(v)$ incident
edges, and each is involved in exactly $n^{1/3}$
triads.
\end{proof}

Lemmas~\ref{le:subgraphedge}--\ref{le:sendmessage} show that the message volume into/out of every vertex is close to its expectation.  By applying Theorem~\ref{th:mixing} and Lemma~\ref{le:knowaboutlabel}, all messages can be routed in $O(n^{1/3+o(1)})$ time.  This concludes the proof of Theorem~\ref{le:mainlemma}.

\begin{cor}\label{cor:triangle}
Let $G$ be a graph with $\mix(G) = n^{o(1)}$.
In the $\CONGEST$ model,
Triangle Detection, Enumeration, and Counting can be solved on $G$, with high probability,
in $n^{1/3+o(1)}$ time.
\end{cor}


\subsection{Triangle Enumeration and Counting in General Graphs}

The algorithm for Theorem~\ref{th:result1} is based on an $n^{1/2}$-decomposition.
Since the connected components induced by $E_m$ have low mixing time, we can solve Triangle Enumeration/Counting on 
them very efficiently using Theorem~\ref{le:mainlemma},
in $n^{1/3+o(1)}$ time, i.e., much less than the time 
required to compute the $n^{1/2}$-decomposition.

\begin{theorem}\label{th:result1}
In the $\CONGEST$ model, Triangle Detection, Counting, and Enumeration
can be solved, w.h.p., in $\tilde{O}(n^{1/2})$ rounds.
\end{theorem}
\begin{proof}
The underlying graph is $G=(V,E)$. We set the parameter $\delta=1/2$. By Theorem~\ref{th:main}, we compute an
$\ndel$-decomposition $E=E_m\cup E_s \cup E_r$ using $\tilde{O}(n^{1-\delta})$ rounds. 
We divide the task of enumerating triangles into three
cases.  By ensuring that every triangle is output 
by exactly one vertex, this algorithm also solves Triangle \emph{Counting}.

\paragraph{Case 1: All Triangles Intersecting $E_s$.}
We handle this case as follows. By Condition (b) of Definition~\ref{def:decomp}, 
$E_s=\bigcup_{v\in V} E_{s,v}$, 
where $\{E_{s,v}\}$ defines an acyclic $n^\delta$-orientation.
We let each $v$ announce $E_{s,v}$ to all its neighbors,
in $O(n^{\delta})$ time. 
For the Triangle Counting application
it is important that every triangle $\{x,y,z\}$ intersecting $E_s$
be reported by exactly \emph{one} vertex.
If $(x,y)$ and $(x,z)$ are oriented and $\ID(y)<\ID(z)$, 
then $y$ detects and reports the triangle.  
If $(x,z)$ is oriented, 
$\{x,y\}$ is unoriented, 
and $\{y,z\}$ is unoriented or
oriented as $(z,y)$, then $y$
detects and reports the triangle.  
If $(x,z), (y,z)$ are oriented but $\{x,y\}$ is not, 
and $\ID(y)<\ID(x)$, $y$ reports the triangle.\footnote{Most of these cases
do not occur in the parital orientations produced by our algorithm; 
nonetheless, they can occur in arbitrary partial acyclic orientations.}

\paragraph{Case 2: Some Triangles Intersecting $E_m$.}
Consider a single connected  component $\Gin=(\Vin,\Ein)$ 
induced by $E_m$, which has mixing time $n^{o(1)}$.
We classify vertices in $\Vin$ as \emph{good} or \emph{bad}
depending on whether they naturally satisfy Condition (i)
of Theorem~\ref{le:mainlemma}.  A vertex is good
if $\degin(v) \ge \deg_{E_r}(v)$.  Let $\Eout$
be the subset of $E_r$-edges incident to \emph{good} vertices in $\Vin$,
and let $E_r^{\text{new}}$ be
the subset of $E_m$-edges incident to \emph{bad} vertices in $\Vin$.
We now apply Theorem~\ref{le:mainlemma}
to enumerate/count all triangles in the edge set $\Ein\cup\Eout$.
(Because triangles contained in $E_r^{\text{new}}$ will also 
be found in Case 3, the Triangle \emph{Counting} algorithm 
should refrain from including these in the tally for Case 2.)

\paragraph{Case 3: Triangles Contained in $E_r^{\operatorname{new}}\cup E_r$.}
Let $E_r^{\text{new}}$ denote the union of all such sets, 
over all components in $E_m$.  Since each edge in $E_r^{\text{new}}$
can be charged to an endpoint of an edge in $E_r$,
$|E_r^{\text{new}} \cup E_r| \leq 3|E_r| = |E|/2$.
We apply the algorithm recursively to the graph 
induced by $E_r^{\text{new}} \cup E_r$.  The depth of the recursion
is obviously at most $\log_2 m$.

\paragraph{Round Complexity.}
Computing an $\ndel$-decomposition $E=E_m\cup E_s \cup E_r$ takes $\tilde{O}(n^{1-\delta})$ rounds. 
The algorithm for Case~1 takes $O(\ndel)$ rounds.
The algorithm for Case~2 takes $O(n^{1/3+o(1)})$ rounds.
The number of recursive calls (Case~3) is $\log m$.
Thus, the overall round complexity 
is 
\[
\log m \cdot  \left( {O}(\ndel)+ \tilde{O}(n^{1-\delta})+O(n^{1/3+o(1)}) \right) = \tilde{O}(n^{1/2}).\qedhere
\]
\end{proof}

\subsection{Subgraph Enumeration}
In this section we show that Corollary~\ref{cor:triangle} can be extended to enumerating  $s$-vertex subgraphs in $O(n^{(s-2)/s+o(1)})$ rounds. Note that the $\Omega(n^{1/3}/\log n)$ lower bound for triangle enumeration on \Erdos-\Renyi{} graphs
$\mathcal{G}(n, 1/2)$~\cite{IzumiL17} can be generalized to an $\Omega(n^{(s-2)/s}/\log n)$ lower bound for enumerating $s$-vertex subgraphs.  This implies that Theorem~\ref{th:subgraphlisting} is nearly optimal on 
$\mathcal{G}(n, 1/2)$.

\begin{theorem}\label{th:subgraphlisting}
Let $s = O(1)$ be any constant. Given a graph $G$ of $n$ vertices with $\mix(G) = n^{o(1)}$, we can list all $s$-vertex subgraphs of $G$ in $O(n^{(s-2)/s+o(1)})$ rounds, w.h.p., in the $\CONGEST$ model.
\end{theorem}

This theorem is proved using a variant of Lemma~\ref{le:mainlemma}, also with $\mathcal{G}=G$ and $|\NG(v)|=\degree(v)$. 
The proof of Theorem~\ref{th:subgraphlisting} is almost the same as that of Lemma~\ref{le:mainlemma}, and so in what follows we only highlight the difference.

Similarly, we assume the maximum degree is at most $m/\left(20n^{(s-2)/s}\log n\right)$, since otherwise we can apply Theorem~\ref{th:mixing} to let $v$ learn the entire set $E$ in $O(n^{(s-2)/s+o(1)})$ rounds, and we are done after that. 

We partition $V$ into $n^{1/s}$ subsets $V_1,\ldots,V_{n^{1/s}}$. Instead of considering triads, here we consider $s$-tuples: 
$\left\{(i_1, \ldots ,i_s) \mid 1 \le i_1 \le  \ldots \le i_s \le n^{1/s}\right\}$.
After a vertex $v$ learns the entire edge set $\cup_{j_1,j_2\in[1,s]}E(V_{i_{j_1}},V_{i_{j_2}})$, it has ability to list all $s$-vertex subgraphs in which the $j$th vertex is in $V_{i_j}$.
We prove a variant of Lemma~\ref{le:edgenumber}, as follows.

\begin{lemma}\label{le:subgraph_edgenumber}
W.h.p., $|E(V_i,V_j)|\leq 24m/n^{2/s}$ for all $i,j\in[1,n^{1/s}]$.
\end{lemma}
\begin{proof}
We set $p=n^{-1/s}$. The maximum degree is at most $m/\left(20n^{(s-2)/s}\log n\right)\leq  mp/20\log n$, and $p^2m\geq m/n^{2/s}\geq 400\log^2n$. By applying Lemma~\ref{le:subgraphedge} and use the same analysis in Lemma~\ref{le:edgenumber}, we  conclude this lemma.
\end{proof}

\begin{proof}[Proof of Theorem~\ref{th:subgraphlisting}]
Here we only consider the time complexity to delivery all messages.
Consider a vertex $v$. If $k_v<1/2$, then $v$ receives no message.  Otherwise $v$ is responsible for between $2k_v$ and $4k_v$ tuples, 
and $v$ collects $E(V_i,V_j)$ for at most $4s^2k_v$ pairs $(V_i, V_j)$. By Lemma~\ref{le:subgraph_edgenumber}, 
w.h.p., $|E(V_i,V_j)|\leq 24m/n^{2/s}$ for all $i,j$. Hence the number of edges $v$ received is at most $24m/n^{2/s}\cdot 4s^2k_v=O(\degree(v)\cdot n^{(s-2)/s})$. 

Note that each vertex $v$ sends at most $O(\degree(v)n^{(s-2)/s})$ messages since for each incident edge $e$ of $v$, there are at most $O(s^2n^{(s-2)/s})$ tuples involving $e$. 
By Theorem~\ref{th:mixing}, the delivery of all messages can be done in $O(n^{(s-2)/s+o(1)})$ rounds, w.h.p.
\end{proof}

\section{Conclusion}\label{sect:conclusion}

In this paper we have shown that all variants of Triangle Detection, Enumeration, and Counting can be solved in $\tilde{O}(n^{1/2})$ rounds in the $\CONGEST$ model.  Moreover, we have shown that with better distributed graph partitioning technology, our algorithm can be implemented in $O(n^{1/3+o(1)})$ rounds, nearly achieving the $\Omega(n^{1/3}/\log n)$ $\CONGEST$
lower bound of Izumi and LeGall~\cite{IzumiL17}.\footnote{Of course, their lower bound holds for dense graph with $\Theta(n^2)$ edges; it is unclear whether $\tilde{\Omega}(n^{1/3})$ is a valid lower bound for all possible $\Delta$ in the relevant range $n^{1/3} < \Delta < n$.
The lower bound clearly does not hold when $\Delta \ll n^{1/3}$.}
Thus, the main question left open by our work is whether such an ideal graph decomposition can be computed efficiently,
in $O(n^{1/3+o(1)})$ time or even $O(\polylog(n))$ time.
We conjecture that this is, in fact, possible.

\begin{conj}\label{conjecture}
In the $\CONGEST$ model, the edge set $E$ can be partitioned
into $E = E_m \cup E_r$ in $\polylog(n)$ time such that
$|E_r| < |E|/2$ and the components induced by $E_m$
have conductance $\Omega(1/\polylog(n))$ and hence
$O(\polylog(n))$ mixing time.
\end{conj}

If Conjecture~\ref{conjecture} is true, then the upper bound for triangle enumeration, detection, and counting is roughly
\[
\min\left\{\, \tilde{O}(\Delta), \;\, n^{1/3+o(1)}\, \right\},
\]
i.e., depending on the magnitude of $\Delta$,
we should execute one of two algorithms.
For which edge-densities and which problems (counting, detection, enumeration) can this upper bound be improved?
Is there a \emph{third} algorithm that is substantially
superior to these two in some regime?

\bibliographystyle{abbrv}
\bibliography{ref}

\end{document}